\documentclass[12pt]{article}

\usepackage{graphicx}
\usepackage{amsmath}
\usepackage{amssymb}
\usepackage{amsthm}
\usepackage{latexsym}
\usepackage{mathrsfs}
\usepackage{enumerate}
\usepackage{ifthen}

\theoremstyle{definition}

\newtheorem{axiom}{Axiom}
\newtheorem{theorem}{Theorem}

\newcommand{\stateset}{\mbox{$\mathscr{S}$}}
\newcommand{\atomicset}{\mbox{$\mathscr{A}$}}
\newcommand{\eidoset}{\mbox{$\mathscr{E}$}}
\newcommand{\recordset}{\mbox{$\mathscr{R}$}}
\newcommand{\infoset}{\mbox{$\mathscr{I}$}}
\newcommand{\mechset}{\mbox{$\mathscr{M}$}}
\newcommand{\procset}{\mbox{$\mathcal{P}$}}
\newcommand{\uniformset}{\mbox{$\mathcal{U}$}}

\newcommand{\absolute}[1]{\left | #1 \right |}
\newcommand{\numberin}[1]{\#(#1)}
\newcommand{\zeroprocess}{\mbox{\bf 0}}
\newcommand{\fproc}[2]{\left \langle #1,#2 \right \rangle}

\newcommand{\eqclass}[1]{\left [ #1 \right ]}
\newcommand{\irrev}{\mathbb{I}}

\newcommand{\entropy}{\mathbb{S}}
\newcommand{\bitstate}{I_{\mathrm{b}}}
\newcommand{\bitproc}{\Theta_{\mathrm{b}}}

\newcommand{\hilbert}{\mathcal{H}}
\newcommand{\ket}[1]{\left | #1 \right \rangle}
\newcommand{\bra}[1]{\left \langle #1 \right |}
\newcommand{\proj}[1]{\ket{#1} \!\! \bra{#1}}
\newcommand{\amp}[2]{\left \langle #1 | #2 \right \rangle}
\newcommand{\oper}[1]{\boldsymbol{#1}}
\newcommand{\tr}{\mbox{Tr}\,}

\newcommand{\leftorright}{\rightleftharpoons}

\title{Axiomatic information thermodynamics}
\date{\today}

\author{Austin Hulse \\ 
	Benjamin Schumacher\thanks{Corresponding author:  Department of Physics,
	Kenyon College, Gambier, OH 43022 USA. E-mail schumacherb@kenyon.edu} 
	\\ Department of Physics, Kenyon College\\
	\and Michael D. Westmoreland\\Department of Mathematics, Denison University}

\begin{document}

\maketitle

\begin{abstract}
	We present an axiomatic framework for thermodynamics that incorporates
	information as a fundamental concept.  The axioms describe both ordinary 
	thermodynamic processes and those in which information is acquired, used 
	and erased, as in the operation of Maxwell's demon.
	This system, similar to previous axiomatic systems for thermodynamics, supports
	the construction of conserved quantities and an entropy function
	governing state changes.  Here, however, the entropy exhibits 
	both information and thermodynamic aspects.
	Although our axioms are not based upon probabilistic concepts,
	a natural and highly useful concept of probability 
	emerges from the entropy function itself.  
	Our abstract system has many models, including both classical 
	and quantum examples.
\end{abstract}

\section{Introduction}  \label{sec:intro}

Axiomatic approaches to physics are useful for exploring
the conceptual basis and logical structure of a theory.  One classic
example was presented by Robin Giles over fifty years ago in his 
monograph \emph{Mathematical Foundations of Thermodynamics}
\cite{Giles1964}.
His theory is constructed upon three phenomenological concepts:  
thermodynamic states, an operation ($+$) that combines states 
into composites, and a relation~($\rightarrow$) describing possible
state transformations.  From a small number of basic
axioms, Giles derives a remarkable amount of thermodynamic 
machinery, including conserved quantities (``components of content''), 
the existence of an entropy function that characterizes 
irreversibility for possible processes, and so on.

Alternative axiomatic developments of thermodynamics have been
constructed by others along different lines.
One notable example is the framework of Lieb and Yngvason
\cite{LiebYngvason1999,LiebYngvason1998}
(which has recently been used by Thess as the basis for a 
textbook \cite{Thess2011}).  Giles's abstract system, meanwhile, has found application
beyond the realm of classical thermodynamics, e.g., in the theory
of quantum entanglement~\cite{VedralKashefi2002}.

Other work on the foundations of thermodynamics
has focused on the concept of information.  
Much of this has been inspired by Maxwell's famous
thought-experiment of the demon \cite{LeffRex2003}.  The demon is an
entity that can acquire and use information about the
microscopic state of a thermodynamic system, producing
apparent violations of the Second Law of Thermodynamics.
These violations are only ``apparent'' because the demon
is itself a physical system, and its information processes
are also governed by the underlying dynamical laws.

Let us examine a highly simplified example of the demon
at work.  Our thermodynamic system is a one-particle
gas enclosed in a container (a simple system also used 
in \cite{Szilard1929}).  The gas may be allowed to freely
expand into a larger volume, but this process is
irreversible.  A ``free compression'' process that took
the gas from a larger to a smaller volume with no other
net change would decrease the entropy of the system 
and contradict the Second Law.
See Figure~\ref{fig:nofreecompression}.
\begin{figure}
\begin{center}
\includegraphics[width=3.5in]{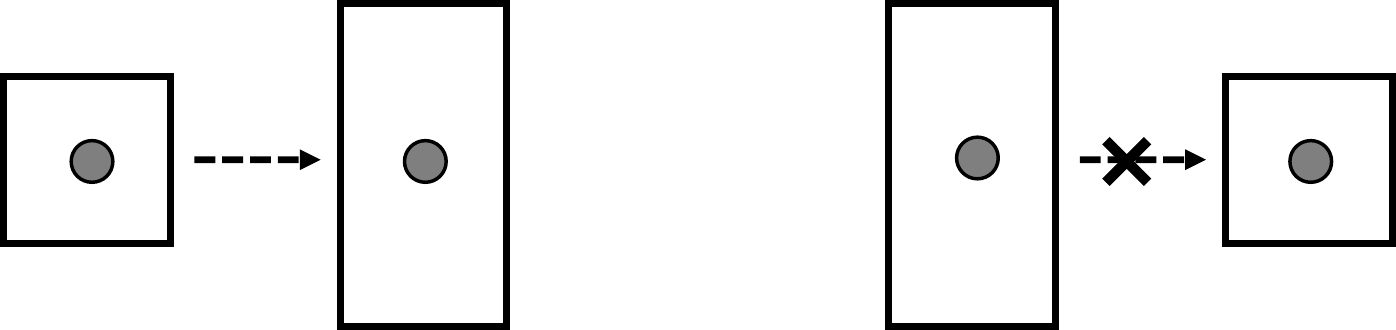}
	\end{center}
	\caption{\label{fig:nofreecompression} A one-particle gas
		may freely expand to occupy a larger volume, but the 
		reverse process would violate the Second Law.}
\end{figure}

Now, we introduce the demon, which is a machine that can
interact with the gas particle and acquire information about
its location.  The demon contains a one-bit
memory register, initially in the state 0.  First, a partition is 
introduced in the container, so that the gas particle is confined
to the upper or lower half ($U$ or $L$).  The demon measures the
gas particle and records its location in memory, with 0 standing 
for $U$ and 1 for $L$.  On the basis of this memory 
bit value, the demon moves
the volume containing the particle.  In the end, the gas 
particle is confined to one particular smaller volume, apparently
reducing the entropy of the system.  
This is illustrated in Figure~\ref{fig:maxwellsdemon}.
\begin{figure}
\begin{center}
\includegraphics[width=4.5in]{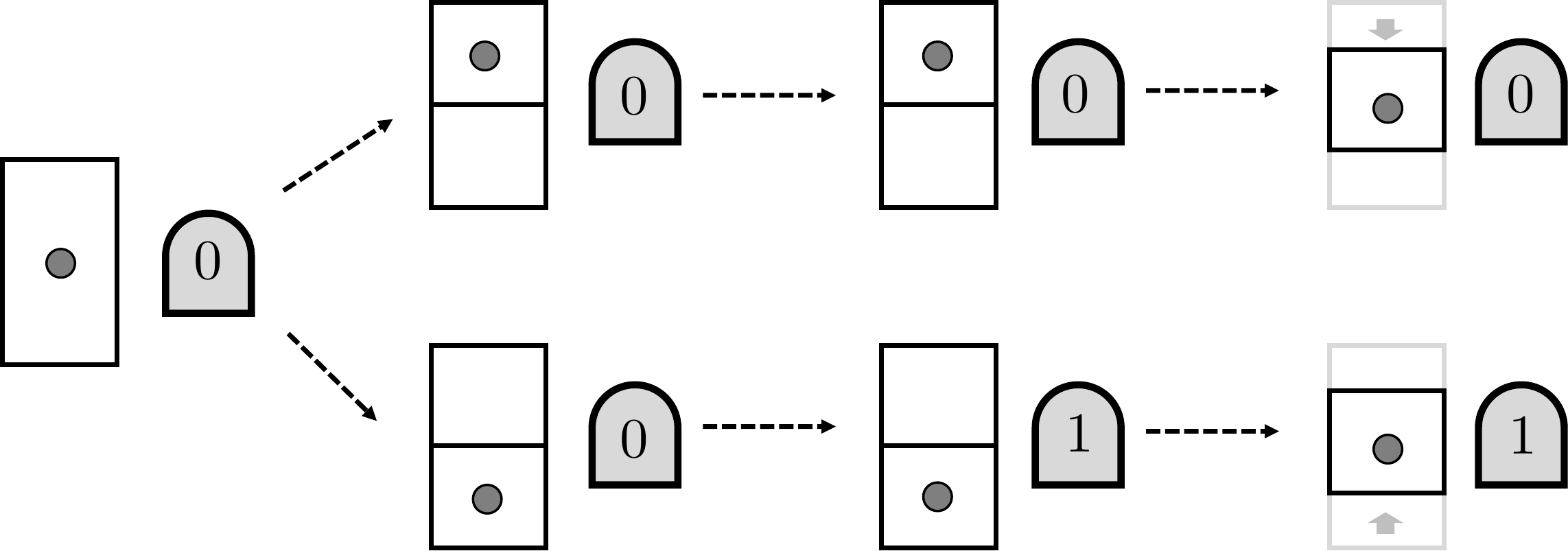}
	\end{center}
	\caption{\label{fig:maxwellsdemon} A Maxwell's demon device
		acquires one bit of information about the location of a gas
		particle, allowing it to contract the gas to a smaller volume.}
\end{figure}

As Bennett pointed out \cite{Bennett1982}, every demon operation we have described 
can be carried out reversibly, with no increase in global entropy.
Even the ``measurement'' process can be described by a reversible
interaction between the particle location and the demon's 
memory bit $b$, changing the states according~to

\begin{equation}
	( \mbox{U}, b ) \dashleftarrow \dashrightarrow ( \mbox{U}, b )
	\qquad
	( \mbox{L}, b ) \dashleftarrow \dashrightarrow ( \mbox{L}, \bar{b} ),
\end{equation}
where $\bar{b}$ is the binary negation of $b$.  However, the
demon process as described leaves the demon in a new
situation, since an
initially blank memory register now stores a bit of information.
To operate in a cycle (and thus unambiguously violate the
Second Law), the demon must erase this bit
and recover the blank memory.  However, as Landauer 
showed \cite{Landauer1961}, the erasure of a bit of information is 
always accompanied by an entropy increase of at 
least $k_{B} \ln 2$ in the surroundings, just enough to
ensure that there is no overall entropy decrease in the
demon's operation.  The Second Law remains valid.

Information erasure is a physical process.
The demon we have described can also erase a stored bit
simply by reversing the steps outlined.  Depending
on the value of the bit register, the demon moves the gas particle
to one of two corresponding regions separated by a partition.
The same interaction that accomplished the ``measurement'' 
process (which takes $(U,0) \dashrightarrow (U,0)$ and $(L,0) \dashrightarrow (L,1)$)
can now be used to reset the bit value to 0
(taking $(U,0) \dashrightarrow (U,0)$ and $(L,1) \dashrightarrow (L,0)$).  
In other words, the information stored redundantly in both the gas and 
the register is ``uncopied'', so that it remains only in the gas.  Finally, the 
partition is removed and the gas expands to the larger volume.
The net effect is to erase the memory register while increasing
the entropy of the gas by an amount $k_{B} \ln 2$, in conventional units.

Another link between thermodynamics and information comes from
statistical mechanics.  As Jaynes has shown \cite{Jaynes1957a,Jaynes1957b}, 
the concepts and measures of information devised by
Shannon for communication theory \cite{Shannon1948} can be used in
the statistical derivation of macroscopic thermodynamic properties.
In a macroscopic system, we typically possess only a small amount
of information about a few large-scale parameters (total
energy, volume, etc.).  According to Jaynes, we should
therefore choose a probability distribution over microstates
that maximizes the Shannon entropy, consistent with our 
data.  That is to say, the rational probability assignment 
includes no information (in the Shannon sense) except that 
found in the macroscopic state of the system.  This 
prescription yields the usual distributions used in
statistical mechanics, from which thermodynamic properties
may be derived.

Axiomatic theories and information analyses each
provide important insights into the meaning of thermodynamics.
The purpose of this paper is to synthesize these
two approaches.  We will present an axiomatic basis for 
thermodynamics that uses information ideas from the very
outset.  In such a theory, Maxwell's demon, which accomplishes
state changes by acquiring information, is neither a paradox nor
a sideshow curiosity.  Instead, it is a central conceptual tool 
for understanding the transformations between thermodynamic states.  
In our view, thermodynamics is essentially a theory of the descriptions
of systems possessed by agents that are themselves (like the
demon) physical systems.  These descriptions may change as
the external systems undergo various processes; however, they also
may change when the agent acquires, uses or discards
information.

Our work is thus similar in spirit to that of Weilenmann et al. \cite{WeilenmannRenner2016},
though our approach is very different.  They essentially take the 
Lieb-Yngvason axiomatic framework and apply it to various quantum
resource theories.  In particular, for the resource theory of
non-uniformity \cite{SpekkensEtc2015}, the Lieb-Yngvason
entropy function coincides with the von Neumann entropy of
the quantum density operator, which is a measure of quantum 
information.  We, on the other hand, seek to modify 
axiomatic thermodynamics itself to describe processes
involving ``information engines'' such as Maxwell's demon.  We
do not rely on any particular microphysics and have
both classical and quantum models for our axioms.  The
connections we find between information and thermodynamic
entropy are thus as general as the axioms themselves.  
Furthermore, to make our concept of information as clear as possible, 
we will seek to base our development on the most elementary 
ideas of state and process.

We therefore take as our prototype the axiomatic theory of Giles \cite{Giles1964}.
In fact, Giles's monograph contains two different axiomatic 
developments. (The first system is presented in 
Chapters 1--6 of Giles's book, and the second is introduced
beginning in Chapter 7.)
The first, which we might designate Giles~I, 
is based on straightforward postulates about the properties
of states and processes.  The second, Giles II, is more sophisticated
and powerful.  The axioms are less transparent in meaning (e.g., the
assumed existence of ``anti-equilibrium'' and ``internal'' 
thermodynamic states), but they support stronger theorems.  
This difference can be illustrated by the status of entropy
functions in the two developments.  In Giles I, it is shown
that an entropy function exists; in fact, there may be many such functions.
In Giles II, it is shown that an \emph{absolute} entropy function,
one that is zero for every anti-equilibrium state, exists and
is unique.
Giles himself regarded the second system as ``The 
Formal Theory'', which he summarizes in an Appendix of
that title. 

The information-based system we present here,
on the other hand, is closer to the more elementary
framework of Giles I.
We have taken some care to use notation and 
concepts that are as analogous as possible to that theory.  
We derive many of Giles's propositions, 
including some that he uses as axioms.  
Despite the similarities, however,
even the most elementary ideas (such as the combination of 
states represented by the $+$ operation) will require 
some modifications.
These changes highlight the new ideas embodied
in axiomatic information thermodynamics, and so
we will take some care to discuss them as they arise.

Section~\ref{sec:eidostates} introduces the fundamental
idea of an eidostate, including how two or more eidostates
may be combined.  In Section~\ref{sec:processes}, we
introduce the $\rightarrow$ relation between eidostates:
$A \rightarrow B$ means that eidostate $A$ can be 
transformed into eidostate $B$, with no other net change
in the apparatus that accomplishes the transformation.  
The collection of processes has an algebraic structure, which we present
in Section~\ref{sec:algebra}.  We introduce our concept
of information in Section~\ref{sec:information} and show
that the structure of information processes imposes a
unique entropy measure on pure information states.

Section~\ref{sec:demons} introduces axioms inspired by
Maxwell's demon and shows their implications for 
processes involving thermodynamic states. 
Sections~\ref{sec:irreversibility} and \ref{sec:mechanical}
outline how our information axioms also yield
Giles's central results about thermodynamic processes,
including the existence of an entropy function, conserved
components of content, and mechanical states.
In Section~\ref{sec:uniformentropy}, we see that an 
entropy function can be extended in a unique way to
a wider class of ``uniform'' eidostates.  This, rather
remarkably, gives rise to a unique probability measure 
within uniform eidostates, as we describe
in Section~\ref{sec:probability}.

Sections~\ref{sec:model1} and \ref{sec:model2} present two detailed models
of our axioms, each one highlighting a different aspect of
the theory.  We conclude in Section~\ref{sec:remarks}
with some remarks on the connection between information
and thermodynamic entropy, a connection that emerges
necessarily from the structure of processes in our theory.

\section{Eidostates}  \label{sec:eidostates}

Giles \cite{Giles1964}
bases his development on a set $\stateset$ of \emph{states}.  
The term ``state'' is undefined in the formal theory, but heuristically 
it represents an equilibrium macrostate of some thermodynamic system.
(Giles, in fact, identifies $a \in \stateset$ with some method of 
preparation; the concept of a ``system'' does not appear in his theory, or in ours.)
States can be combined using the $+$ operation, so that if $a,b \in \stateset$
then $a+b \in \stateset$ also.  The new state $a+b$ is understood as 
the state of affairs arising from simultaneous and
independent preparations of states $a$ and $b$.  For Giles, this
operation is commutative and associative; for example, $a+b$
is exactly the same state as $b+a$.

We also have a set $\stateset$.  However, our theory differs in
two major respects.  First, the combination of states is not 
assumed to be commutative and associative.  On the contrary,
we regard $a+b$ and $b+a$ as entirely distinct states for any $a \neq b$.
The motivation for this is that the way that a composite state is
assembled encodes information about its preparation, and
we want to be able to keep track of this information.  On the
states in $\stateset$, the operation $+$ is simply that of forming
a Cartesian pair of states:  $a+b = (a,b)$, and nothing more.

A second and more far-reaching difference is that we cannot
confine our theory to individual elements of $\stateset$.  A 
theory that keeps track of information in a physical system
must admit non-deterministic processes.  For instance, if
a measurement is made, the single state $a$ before 
the measurement may result in any one of a collection of 
possible states $\{ a_0, a_1, \ldots , a_n \}$ corresponding to the 
various results.

Therefore, our basic notion is the \emph{eidostate},
which is a finite nonempty set of states. The
term ``eidostate'' derives from the Greek word
\emph{eidos}, meaning ``to see''.  The eidostate
is a collection of states that may be regarded as
possible from the point of view of some agent. The set of eidostates 
is designated by $\eidoset$.  When we combine
two eidostates $A,B \in \eidoset$ into the 
composite $A + B$, we mean the set of all
combinations $a+b$ of elements of these sets.
That is, $A + B$ is exactly the Cartesian product
(commonly denoted $A \times B$) 
for the sets. Thus, our
state combination operation $+$ is very different
from that envisioned by Giles.  His operation 
is an additional structure on $\stateset$ that
requires an axiom to fix its properties,
whereas we simply use the Cartesian product
provided by standard set theory, which we of course
assume.

Some eidostates in $\eidoset$ are Cartesian products 
of other sets (which are also eidostates); some eidostates
are not, and are thus ``prime''.  We assume that each eidostate
has a ``prime factorization'' into a finite number of components.
Here, is our first axiom:
\begin{axiom}Eidostates:  \label{axiom:eidostates}
	$\eidoset$ is a collection of sets called \emph{eidostates} such that:
	\begin{description}
		\item(a)  Every $A \in \eidoset$ is a finite nonempty set
			with a finite prime Cartesian factorization.
		\item(b)  $A + B \in \eidoset$ if and only if $A,B \in \eidoset$.
		\item(c)  Every nonempty subset of an eidostate is also an eidostate.
	\end{description}
\end{axiom}
Part (c) of this axiom ensures, among other things, that every element
$a$ of an eidostate can be associated with a singleton eidostate $\{ a \}$.
Without too much confusion, we can simply denote any singleton eidostate
$\{ a \}$ by the element it contains, writing $a$ instead.  We can therefore regard
the set of states $\stateset$ in two ways.  Either we may think of $\stateset$
as the collection of singleton eidostates in $\eidoset$, or we may say that
$\stateset$ is the collection of all elements of all eidostates:  
$\displaystyle \stateset = \bigcup_{A \in \eidoset} A$.  Either way, the
set $\eidoset$ characterized by Axiom~\ref{axiom:eidostates} is the
more fundamental object.

Our ``eidostates'' are very similar to the ``specifications''
introduced by del Rio et al. \cite{delRioKraemerRenner} in their
general framework for resource theories of knowledge.  The
two ideas, however, are not quite identical.  
To begin with, a specification $V$ may be any
subset of a state space $\Omega$, whereas the eidostates in $\eidoset$
are required to be finite nonempty subsets of $\stateset$---and indeed, not
every such subset need be an eidostate.  For example, the union
$A \cup B$ of two eidostates is not necessarily an eidostate.
The set $\eidoset$ therefore does not form a Boolean lattice.  
Specifications are a general concept applicable
to state spaces of many different kinds, and are used in 
\cite{delRioKraemerRenner} to analyze different resource
theories and to express general notions of approximation and locality.
We, however, will be restricting our attention to an $\eidoset$ (and $\stateset$)
with very particular properties, expressed by Axiom~\ref{axiom:eidostates}
and our later axioms, that are designed to model thermodynamics 
in the presence of information engines such as Maxwell's demon.

Composite eidostates are formed by combining eidostates with the
$+$ operation (Cartesian product).  The same pieces may be combined
in different ways.  We say that $A,B \in \eidoset$ are \emph{similar}
(written $A \sim B$) if they are made up of the same components.
That is, $A \sim B$ provided there are eidostates $E_{1}, \ldots , E_{n}$
such that $A = F_{A}(E_{1}, \ldots , E_{n})$ and $B = F_{B}(E_{1}, \ldots , E_{n})$
for two Cartesian product formulas $F_{A}$ and $F_{B}$.  Thus, 

\begin{equation}
	(E_{1} + E_{2}) + E_{3} \sim E_{2} + (E_{1} + E_{3}) ,
\end{equation}
and so on.  The similarity relation $\sim$ is an equivalence relation 
on $\eidoset$.

We sometimes wish to combine an eidostate with itself several
times.  For the integer $n \geq 1$, we denote by $n A$ 
the eidostate $A + (A + \ldots)$, where $A$ appears $n$ times
in the nested Cartesian product.  This is one particular way to 
combine the $n$ instances of $A$, though of course all
such ways are similar.  Thus, we may assert equality in
$A + nA = (n+1) A$, but only similarity in $nA + A \sim (n+1) A$.

Finally, we note that we have introduced as yet no probabilistic
ideas.  An eidostate is a simple enumeration of possible states, 
without any indication that some are more or less likely than others.
However, as we will see in Section~\ref{sec:probability}, 
in certain contexts, a natural probability measure 
for states does emerge from our axioms.

\section{Processes}  \label{sec:processes}

In the axiomatic thermodynamics of Giles, the $\rightarrow$ relation
describes state transformations.  The relation $a \rightarrow b$
means that there exists another state $z$ and a definite time
interval $\tau \geq 0$ so that

\begin{equation}
	a + z \stackrel{\tau}{\rightsquigarrow} b + z ,
\end{equation}
where $\stackrel{\tau}{\rightsquigarrow}$ indicates time evolution
over the period $\tau$.  The pair $(z,\tau)$ is the ``apparatus''
that accomplishes the transformation from $a$ to $b$.  
This dynamical evolution is a deterministic process; 
that is, in the presence of the apparatus $(z,\tau)$ 
the initial state $a$ is guaranteed to
evolve to the final state $b$.  This rule of interpretation for $\rightarrow$
motivates the properties assumed for the relation in the axioms.

Our version of the arrow relation $\rightarrow$ is slightly 
different, in that it encompasses non-deterministic processes.  
Again, we envision an apparatus $(z,\tau)$ and we write 

\begin{equation}
	a + z \stackrel{\tau}{\dashrightarrow} b + z ,
\end{equation}
to mean that the initial state $a + z$ may possibly
evolve to $b + z$ over the stated interval.  Then,
for eidostates $A$ and $B$, the relation $A \rightarrow B$
means that, if $a \in A$ and $b \in \stateset$, 
then there exists an apparatus $(z,\tau)$ such that
$a + z \stackrel{\tau}{\dashrightarrow} b + z$ only if
$b \in B$.  Each possible initial state $a$ in $A$ might evolve to
one or more final states, but all of the possible final 
states are contained in $B$.  For singleton eidostates $a$ and $b$, the
relation $a \rightarrow b$ represents a deterministic
process, as in Giles's theory.

Again, we use our heuristic interpretation of $\rightarrow$ to 
motivate the essential properties specified in an axiom:
\begin{axiom}Processes:  \label{axiom:processes}
Let eidostates $A,B,C \in \eidoset$, and $s \in \stateset$.
\begin{description}
	\item(a)  If $A \sim B$, then $A \rightarrow B$.
	\item(b)  If $A \rightarrow B$ and $B \rightarrow C$, then $A \rightarrow C$.
	\item(c)  If $A \rightarrow B$, then $A + C \rightarrow B + C$.
	\item(d)  If $A + s \rightarrow B + s$, then $A \rightarrow B$.
\end{description}
\end{axiom}
Part (a) of the axiom asserts that it is always possible to ``rearrange the pieces''
of a composite state.  Thus, $A+B \rightarrow B+A$ and so on. 
Of course, since $\sim$ is a symmetric relation, $A \sim B$
implies both $A \rightarrow B$ and $B \rightarrow A$, which
we might write as $A \leftrightarrow B$.

Part (b) says that a process from $A$ to $C$ may proceed via
an intermediate eidostate $B$.  Parts (c) and (d) establish the 
relationship between $\rightarrow$ and $+$.  We can always append
a ``bystander'' state $C$ to any process $A \rightarrow B$, and 
a bystander singleton state $s$ in $A + s \rightarrow B + s$ can 
be viewed as part of the apparatus that accomplishes the 
transformation $A \rightarrow B$.

We use the $\rightarrow$ relation to characterize various
conceivable processes.  A \emph{formal process} is simply a pair 
of eidostates $\fproc{A}{B}$.  Following Giles, we
may say that $\fproc{A}{B}$ is:
\begin{itemize}
	\item  \emph{natural} if $A \rightarrow B$;
	\item  \emph{antinatural} if $B \rightarrow A$;
	\item  \emph{possible} if $A \rightarrow B$ or $B \rightarrow A$ (which may
		be written $A \leftorright B$);
	\item  \emph{impossible} if it is not possible;
	\item  \emph{reversible} if $A \leftrightarrow B$; and
	\item  \emph{irreversible} if it is possible but not reversible.
\end{itemize}
Thus, any formal process must be one of four \emph{types}:
reversible, natural irreversible, antinatural irreversible, or
impossible.

Any nonempty subset of an eidostate is an eidostate.  If the
eidostate $A$ is an enumeration of possible states, the proper subset
eidostate $B \subsetneq A$ may be regarded as an enumeration
with some additional condition present that eliminates one or more
of the possibilities.  What can we say about processes involving
these ``conditional'' eidostates?  To answer this question, we
introduce two further axioms.  The first is this:
\begin{axiom} \label{axiom:austin}
	If $A,B \in \eidoset$ and $B$ is a proper subset of $A$, then $A \nrightarrow B$.
\end{axiom}
This expresses the idea that no natural process can simply
eliminate a state from a list of possibilities.  This is a deeper principle
than it first appears.  Indeed, as we will find, in our theory, it is the
essential ingredient in the Second Law of Thermodynamics.

To express the second new axiom, we must introduce the notion of
a \emph{uniform} eidostate.  The eidostate $A$ is said to be uniform
provided, for every $a,b \in A$, we have $a \leftorright b$.  That
is, every pair of states in $A$ is connected by a possible process.
This means that all of the $A$ states are ``comparable'' in some way
involving the $\rightarrow$ relation.

All singleton eidostates are uniform.  Are there any non-uniform eidostates?
The axioms in fact do not tell us.  We will find models of the axioms
that contain non-uniform eidostates and others that contain none.
Even if there are states $a,b \in \stateset$ for which $a \nrightarrow b$ 
and $b \nrightarrow a$, nothing in our axioms guarantees the existence
of an eidostate that contains both $a$ and $b$ as elements.  We 
denote the set of uniform eidostates by $\uniformset \subseteq \eidoset$.

Now, we may state the second axiom about conditional processes.
\begin{axiom}Conditional processes: \label{axiom:conditional} \hspace{0.5in}
	\begin{description}
		\item(a)  Suppose $A, A' \in \eidoset$ and $b \in \stateset$.  If $A \rightarrow b$
			and $A' \subseteq A$ then $A' \rightarrow b$.
		\item(b)  Suppose $A$ and $B$ are uniform eidostates that are each disjoint
			unions of eidostates: $A = A_{1} \cup A_{2}$ and $B = B_{1} \cup B_{2}$.
			If $A_{1} \rightarrow B_{1}$ and $A_{2} \rightarrow B_{2}$ then
			$A \rightarrow B$.
	\end{description}
\end{axiom}

The first part makes sense given our interpretation of the $\rightarrow$
relation in terms of an apparatus $(z,\tau)$.  If every $a \in A$ satisfies 
$a + z \stackrel{\tau}{\dashrightarrow} b + z$ for only one state $b$, then
the same will be true for every $a \in A' \subseteq A$.  Part (b) of the axiom
posits that, if we can find an apparatus whose dynamics transforms $A_{1}$ states
into $B_{1}$ states, and another whose dynamics transforms $A_{2}$ states
into $B_{2}$ states, then we can devise a apparatus with ``conditional dynamics'' 
that does both tasks, taking $A$ to $B$.    This is a rather strong proposition,
and we limit its scope by restricting it to the special class of uniform eidostates.

We can as a corollary extend Part (b) of Axiom~\ref{axiom:conditional} to more than 
two subsets.  That is, suppose $A,B \in \uniformset$ and the sets $A$ and $B$
are each partitioned into $n$ mutually disjoint, nonempty subsets:  
$A = A_{1} \cup \cdots \cup A_{n}$ and $B = B_{1} \cup \cdots B_{n}$.  
In addition, suppose $A_{k} \rightarrow B_{k}$ for all $k = 1, \ldots , n$. Then,~$A \rightarrow B$.

\section{Process Algebra and Irreversibility}  \label{sec:algebra}

Following Giles, we now explore the algebraic structure 
of formal processes and describe how the type of a possible process
may be characterized by a single real-valued function.  
Although the broad outlines of our development will
follow that of Giles \cite{Giles1964}, there are significant differences.  (For example,
in our theory, the unrestricted set $\procset$ of eidostate 
processes does not form a group.)

There is an equivalence relation among formal processes,
based on the similarity relation $\sim$ among eidostates.
We say that $\fproc{A}{B} \doteq \fproc{C}{D}$ if there
exist singletons $x,y \in \stateset$ such that
$A+x \sim C + y$ and $B + x \sim D + y$.  That is,
if we augment the eidostates in $\fproc{A}{B}$ by $x$
and those in $\fproc{C}{D}$ by $y$, the corresponding eidostates
in the two processes are mere rearrangements of each other.  It is not hard
to establish that this is an equivalence relation.
Furthermore, equivalent processes (under $\doteq$)
are always of the same type.  To see this, suppose that
if $\fproc{A}{B} \doteq \fproc{C}{D}$ and $A \rightarrow B$.
Then, $A + x \sim C + y$, etc.,~and

\begin{equation}
	C + y \rightarrow A + x \rightarrow B + x \rightarrow D + y.
\end{equation}
Hence, $C \rightarrow D$.  It is a straightforward corollary
that $A \rightarrow B$ if and only if $C \rightarrow D$.

Processes (or more strictly, equivalence classes of processes
under $\doteq$) have an algebraic structure.  We define the 
\emph{sum} of two processes as

\begin{equation}
	\fproc{A}{B} + \fproc{C}{D} = \fproc{A+C}{B+D},
\end{equation}
and the \emph{negation} of a process as $- \fproc{A}{B} = \fproc{B}{A}$.
We call the $-$ operation ``negation'' even though in general 
$- \fproc{A}{B}$ is not the additive inverse of $\fproc{B}{A}$.
Such an inverse may not exist.  As we will see, however, the
negation does yield an additive inverse process in some important
special contexts.

The sum and negation operations on processes are both compatible
with the equivalence $\doteq$.  That is, first, 
if $\fproc{A}{B} \doteq \fproc{A'}{B'}$ then $-\fproc{A}{B} \doteq -\fproc{A'}{B'}$.
Furthermore, if $\fproc{A}{B} \doteq \fproc{A'}{B'}$ and $\fproc{C}{D} \doteq \fproc{C'}{D'}$, then

\begin{equation}
	\fproc{A}{B} + \fproc{C}{D} \doteq \fproc{A'}{B'} + \fproc{C'}{D'}.
\end{equation}
This means that the sum and negation operations are well-defined on
equivalence classes of processes.  If $\eqclass{\fproc{A}{B}}$ 
and $\eqclass{\fproc{C}{D}}$ denote two equivalence classes 
represented by $\fproc{A}{B}$ and $\fproc{C}{D}$, then

\begin{equation}
	\eqclass{\fproc{A}{B}} + \eqclass{\fproc{C}{D}} = \eqclass{\fproc{A+C}{B+D}},
\end{equation}
which is the same regardless of the particular representatives chosen to
indicate the two classes being~added.

We let $\procset$ denote the collection of all equivalence classes of formal processes.
Then, the sum operation both is associative and commutative on $\procset$. 
Furthermore, it contains a zero element $\zeroprocess = \eqclass{\fproc{s}{s}}$
for some singleton $s$.  That is, if $\Gamma = \eqclass{\fproc{A}{B}}$,

\begin{equation}
	\Gamma + \zeroprocess = 
	\eqclass{\fproc{A}{B}} + \eqclass{\fproc{s}{s}} = \eqclass{\fproc{A+s}{B+s}}
		= \eqclass{\fproc{A}{B}} = \Gamma.
\end{equation}
The set $\procset$ is thus a monoid (a semigroup with identity) under the
operation $+$.
Moreover, the subset $\procset_{S}$ of singleton processes---equivalence 
classes of formal processes with singleton initial and final states---is 
actually an Abelian group, since 

\begin{equation}
	\eqclass{\fproc{a}{b}} + \eqclass{\fproc{b}{a}}
	= \eqclass{\fproc{a+b}{b+a}} = \eqclass{\fproc{a+b}{a+b}} = \zeroprocess.
\end{equation}
In $\procset_{S}$, the negation operation does yield the additive inverse of an element.

The $\rightarrow$ relation on states induces a corresponding relation on processes.
If $\Gamma, \Delta \in \procset$, we say that $\Gamma \rightarrow \Delta$
provided the process $\Gamma - \Delta$ is natural (a condition we might
formally write as $\Gamma - \Delta \rightarrow \zeroprocess$).  Intuitively,
this means that process $\Gamma$ can ``drive process $\Delta$ backward'',
so that $\Gamma$ and the opposite of $\Delta$ together form a natural
process.

Now, suppose that $\hat{\procset}$ is a collection of equivalence classes of 
possible processes that is closed under addition and negation.  
An \emph{irreversibility function} $\irrev$ is a real-valued function
on $\hat{\procset}$ such that
\begin{enumerate}
	\item  If $\Gamma, \Delta \in \hat{\procset}$, then $\irrev(\Gamma+\Delta)
		= \irrev(\Gamma) + \irrev(\Delta)$.
	\item  The value of $\irrev$ determines the type of the processes in $\hat{\procset}$:
		\begin{itemize}
			\item  $\irrev(\Gamma) > 0$ whenever $\Gamma$ is natural irreversible.
			\item  $\irrev(\Gamma) = 0$ whenever $\Gamma$ is reversible.
			\item  $\irrev(\Gamma) < 0$ whenever $\Gamma$ is antinatural irreversible.
		\end{itemize}
\end{enumerate}
An irreversibility function, if it exists, has a number of elementary
properties.  For instance, since the process $\Gamma + (-\Gamma)$ is
always reversible for any $\Gamma \in \hat{\procset}$, it follows
that $\irrev(-\Gamma) = - \irrev(\Gamma)$.  In fact, we can show that
all irreversibility functions on $\hat{\procset}$ are essentially the same:

\begin{theorem}  \label{theorem:irreversibilityuniqueness}
An irreversibility function on $\hat{\procset}$ is unique up to an overall
positive factor.
\end{theorem}
\begin{proof}
	Suppose $\irrev_{1}$ and $\irrev_{2}$ are irreversibility functions on the set
	$\hat{\procset}$.  If $\Gamma$ is reversible, then $\irrev_{1}(\Gamma) =
	\irrev_{2}(\Gamma) = 0$.  If there are no irreversible processes in $\hat{\procset}$,
	then $\irrev_{1} = \irrev_{2}$.
	
	Now, suppose that $\hat{\procset}$ contains at least one irreversible process $\Gamma$,
	which we may suppose is natural irreversible.  Thus, both $\irrev_{1}(\Gamma) > 0$ and
	$\irrev_{2}(\Gamma) > 0$.  Consider some other process  $\Delta \in \hat{\procset}$.  
	We must show~that
	
	\begin{equation}
		\frac{\irrev_{1}(\Delta)}{\irrev_{1}(\Gamma)} = \frac{\irrev_{2}(\Delta)}{\irrev_{2}(\Gamma)} .
	\end{equation}
	We proceed by contradiction, imagining that the two ratios are not equal and 
	(without loss of generality) that the second one is larger.  Then, there exists a 
	rational number $m/n$ (with $n > 0$) such that
	
	\begin{equation}
		\frac{\irrev_{1}(\Delta)}{\irrev_{1}(\Gamma)} 
			< \frac{m}{n} < \frac{\irrev_{2}(\Delta)}{\irrev_{2}(\Gamma)} .
	\end{equation}
	The first inequality yields $m \irrev_{1}(\Gamma) - n \irrev_{1}(\Delta) > 0$,
	so that the process $m \Gamma - n \Delta$ (that is, $m \Gamma + n (- \Delta)$)
	must be natural irreversible.  The second inequality yields
	$m \irrev_{2}(\Gamma) - n \irrev_{2}(\Delta) < 0$,
	so that the process $m \Gamma - n \Delta$ must be antinatural irreversible.
	These cannot both be true, so the original ratios must be~equal.
\end{proof}

The additive irreversibility function $\irrev$ is only defined on 
a set $\hat{\procset}$ of possible processes.  However, we can 
under some circumstances extend an additive function to
a wider domain.  It is convenient to state here the general 
mathematical result we will use later for this purpose:
\begin{theorem}[``Hahn--Banach theorem'' for Abelian groups.] \label{theorem:hahn-banach}
	Let $\mathcal{G}$ be an Abelian group.  Let $\phi$ be a real-valued
	function defined and additive on a subgroup $\mathcal{G}_{0}$
	of $\mathcal{G}$.  Then, there exists an additive function 
	$\phi'$ defined on $\mathcal{G}$ such that $\phi'(x) = \phi(x)$
	for all $x$ in $\mathcal{G}_{0}$.
\end{theorem}
Note that the extension $\phi'$ is not necessarily unique---that is,
there may be many different extensions of a single additive function
$\phi$.

\section{Information and Entropy}  \label{sec:information}

In our theory, \emph{information} resides in the distinction among
possible states.  Thus, the eidostate $A = \{ a_{1}, a_{2}, a_{3} \}$
represents information in the distinction among it elements.
However, thermodynamic states such as the $a_{k}$s may also
have other properties such as energy, particle content, and so on.
To disentangle the concept of information from the other properties
of these states, we introduce a notion of a ``pure'' information state.

The intuitive idea is this.  We imagine that the world contains 
freely available memory devices.  Different configurations of these
memories---different memory records---are distinct states that 
are degenerate in energy and every other conserved quantity.
Any particular memory record can thus be freely created from or
reset to some null value.

We therefore define a \emph{record state} to be an element 
$r \in \stateset$ such that there exists $a \in \stateset$ so that
$a \leftrightarrow a + r$.  The state $a$ can be thought of as
part of the apparatus that reversibly exchanges the particular
record $r$ with a null value.  In fact, if $A$ is any eidostate
at all, we find that

\begin{equation}
	A + a \leftrightarrow A + (a + r) \leftrightarrow (A+r) + a,
\end{equation}
and so by cancellation of the singleton state $a$, $A \leftrightarrow A + r$.  
We denote the set of record states by $\recordset$.  
If $r,s \in \recordset$, then $r+s \in \recordset$,
and furthermore $r \leftrightarrow s$.  
Any particular record state can be
reversibly transformed into any other,
a fact that expresses the arbitrariness of the ``code''
used to represent information in a memory device.

An \emph{information state} is an eidostate whose elements
are all record states, and the set of such eidostates is
denoted $\infoset$.  All information states are uniform eidostates.
An \emph{information process} is one
that is equivalent to a process $\fproc{I}{J}$,
where $I,J \in \infoset$.  (Of course,
information processes also include processes of the
form $\fproc{I+x}{J+x}$ for a non-record state $x \in \stateset$,
as well as more complex combinations of record and non-record
states.)  Roughly speaking, an information process is a kind
of computation performed on information states.

It is convenient at this point to define a \emph{bit state}
(denoted $\bitstate$) as an information state containing 
exactly two record states.  That is, $\bitstate = \{ r_{0}, r_{1} \}$.
A \emph{bit process} is an information process 
of the form $\bitproc = \fproc{r}{\bitstate}$ for
some $r \in \recordset$---that is, a process by which a
bit state is created from a single record state.

\begin{figure}
\begin{center}
\includegraphics[width=4.5in]{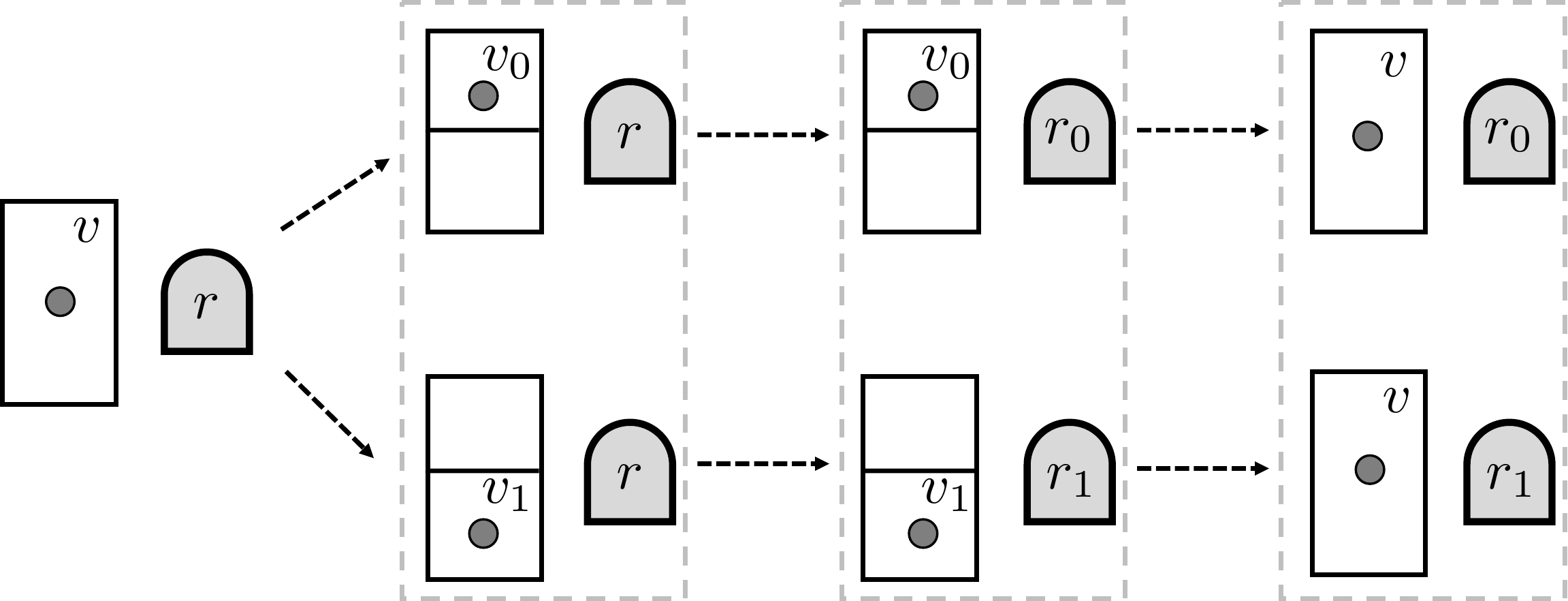}
	\end{center}
	\caption{\label{fig:bitprocess} Maxwell's demon interacting
		with a one-particle gas, illustrating a bit process
		$\fproc{r}{\{r_{0},r_{1}\}}$.}
\end{figure}

We can illustrate a bit process in a thermodynamics context 
by considering a thought-experiment involving Maxwell's demon.
We imagine that the demon operates on a one-particle gas
confined to a volume, 
a situation described by gas state $v$ (see Figure~\ref{fig:bitprocess}).
The demon, whose memory initially has a record state $r$, inserts
a partition into the container, dividing the volume in two
halves labeled 0 and 1.  The gas molecule is certainly in
one sub-volume or the other.  The demon
then records in memory which half the particle occupies.
Finally, the partition is removed and the particle wanders 
freely around the whole volume of the container.  
In thermodynamic terms, the gas state relaxes 
to the original state $v$.  The overall process establishes
the following relations:

\begin{equation}
	v + r \rightarrow \{v_{0} , v_{1} \} + r \rightarrow
		\{ v_{0} + r_{0} , v_{1} + r_{1} \} \rightarrow
		v + \{ r_{0} , r_{1} \},
\end{equation}
and so $r \rightarrow \{r_{0}, r_{1} \}$.  In our
example, the bit process $\bitproc = \fproc{r}{\bitstate}$ is natural one.

However, we do not yet know that there actually \emph{are} 
information states and information processes in our theory.  We
address this by a new axiom.
\begin{axiom}Information: \label{axiom:information}
	There exist a bit state and a possible bit process.
\end{axiom}

Axiom~\ref{axiom:information} has a wealth of consequences.
Since an information state exists, record states necessarily
also exist.  There are infinitely many record states, since
for any $r \in \recordset$ we also have distinct states
$r+r, r+(r+r), \ldots, n r , \ldots,$ all in $\recordset$.  We may have
information states in $\infoset$ that contain arbitrarily many
record states, because $n \bitstate$ contains $2^{n}$ elements.
Furthermore, since every nonempty subset of an information state is also
an information state, for any integer $k \geq 1$ there exists $I \in \infoset$
so that $\numberin{I} = k$.

Any two bit states can be reversibly transformed into one another.  
Consider $\bitstate = \{ r_0,r_1 \}$ and $\bitstate' = \{ s_{0},s_{1} \}$.
Since $r_{0} \rightarrow s_{0}$ and $r_{1} \rightarrow s_{1}$, it
follows by Axiom~\ref{axiom:conditional} that 
$\bitstate \rightarrow \bitstate'$ (and hence $\bitstate 
\leftrightarrow \bitstate'$.)

If any bit process is possible, then every bit process is possible.
Furthermore, every such process is of the natural irreversible type.  
To see why, consider the bit state $\bitstate = \{ r_{0}, r_{1} \}$ and suppose
$\bitstate \rightarrow r$ for a record state $r$.  Since 
$r \rightarrow r_{0}$, this implies that $\{ r_{0}, r_{1} \} \rightarrow r_{0}$,
which is a violation of Axiom~\ref{axiom:austin}.  Therefore,
it must be that $r \rightarrow \bitstate$ but $\bitstate \nrightarrow r$;
and this is true for any choice of $r$ and $\bitstate$.

Now, we can prove that every information process is possible,
and that the $\rightarrow$ relation is determined solely by 
the relative sizes of the initial and final information states.
\begin{theorem}  \label{theorem:infoarrow}
	Suppose $I,J \in \infoset$.  Then, $I \rightarrow J$ if and only if 
	$\numberin{I} \leq \numberin{J}$.
\end{theorem}
\begin{proof}
	To begin with, we can see that $\numberin{I} = \numberin{J}$
	implies $I \rightarrow J$.  This is because we can write
	$I = \{ r_{1},\ldots,r_{n} \}$ and $J = \{s_{1} , \ldots ,s_{n} \}$.
	Since $r_{k} \rightarrow s_{k}$ for every $k = 1,\ldots,n$, the
	finite extension of Axiom~\ref{axiom:conditional} tells us that 
	$I \rightarrow J$.
	
	Now, imagine that $\numberin{I} > \numberin{J}$.  There exists
	a proper subset $I'$ of $I$ so that $\numberin{I'} = \numberin{J}$,
	and thus $J \rightarrow I'$.  If it happened that $I \rightarrow J$,
	it would follow that $I \rightarrow I'$, a contradiction of 
	Axiom~\ref{axiom:austin}.  Hence, $\numberin{I} > \numberin{J}$
	implies $I \nrightarrow J$.
	
	It remains to show that if $\numberin{I} < \numberin{J}$,
	it must be that $I \rightarrow J$.  As a first case, suppose
	that $\numberin{I} = n$ and $\numberin{J} = n+1$.
	Letting $I = \{ r_{1} , \ldots, r_{n} \}$ and $J = \{ s_{1} , \ldots , s_{n}, s_{n+1} \}$,
	we note that $r_{1} \rightarrow s_{1}, r_{2} \rightarrow s_{2}, \ldots,
	r_{n} \rightarrow \{s_{n}, s_{n+1} \}$ (the last being a bit process).
	It follows that $I \rightarrow J$.
	
	We now proceed inductively.  Given $\numberin{J} = \numberin{I} + n$,
	we can imagine a sequence of information states $K_{m}$ with successive
	numbers of elements between $\numberin{I}$ and $\numberin{J}$, so 
	that $\numberin{K_{m}} = \numberin{I} + m$.
	From what we have already proved,
	
	\begin{equation}
		I \rightarrow K_{1} \rightarrow \cdots 
			\rightarrow K_{n-1} \rightarrow J,
	\end{equation}
	and by transitivity of the $\rightarrow$ relation we may conclude that $I \rightarrow J$.
	Therefore, $I \rightarrow J$ if and only if $\numberin{I} \leq \numberin{J}$.
\end{proof}

Intuitively, the greater the number of distinct record states in an information
state, the more information it represents.  Thus, under our axioms, a natural
information process may maintain or increase the amount of information, 
but never decrease it.

We can sharpen this intuition considerably.  Let $\procset_{I}$ be the set of
information processes, which is closed under addition and negation, and
in which every process is possible.  We can define a function
$\irrev$ on $\procset_{I}$ as follows:  For $\Gamma = \eqclass{\fproc{I}{J}} \in \procset_{I}$,

\begin{equation}
	\irrev(\Gamma) = \log \left (  \frac{\numberin{J}}{\numberin{I}}  \right )
		= \log \numberin{J} - \log \numberin{I} .
\end{equation}

The logarithm function guarantees the additivity of $\irrev$ when two processes
are combined, and the sign of $\irrev$ exactly determines whether $I \rightarrow J$.
Thus, $\irrev$ is an irreversibility function on $\procset_{I}$, as the notation suggests.
This function is unique up to a positive constant factor---i.e., the choice of 
logarithm base.  If we choose to use base-2 logarithms, so that a bit process
has $\irrev(\bitproc) = 1$, then $\irrev$ is uniquely determined.

The irreversibility function on information processes can be expressed
in terms of a function on information states.  Suppose we have a 
collection of eidostates $\mathscr{K} \subseteq \eidoset$ that is
closed under the $+$ operation.  With Giles, we define a \emph{quasi-entropy}
on $\mathscr{K}$ to be a real-valued function $\entropy$ such that, 
for $A,B \in \mathscr{K}$:
\begin{enumerate}
	\item  $\entropy(A+B) = \entropy(A) + \entropy(B)$.
	\item  If $\fproc{A}{B}$ is natural irreversible, then $\entropy(A) < \entropy(B)$.
	\item  If $\fproc{A}{B}$ is reversible, then $\entropy(A) = \entropy(B)$.
\end{enumerate}
(A full ``entropy'' function satisfies one additional requirement, which we 
will address in Section~\ref{sec:mechanical} below.)
Given a quasi-entropy $\entropy$, we can derive an irreversibility function
$\irrev$ on possible
$\mathscr{K}$-processes by $\irrev(\fproc{A}{B}) = \entropy(B) - \entropy(A)$.

Obviously, $\entropy(I) = \log \numberin{I}$ is a quasi-entropy function on $\infoset$
that yields the irreversibility function on $\procset_{I}$.  
We recognize it as the Hartley--Shannon entropy of an information
source with $\numberin{I}$ possible outputs \cite{CoverThomas}.
In fact, this is the only possible quasi-entropy on $\infoset$.  
Since every information process is possible,
two different quasi-entropy functions $\entropy$ and $\entropy'$ can only differ
by an additive constant.  However, since $I + r \leftrightarrow I$ for $I \in \infoset$
and $r \in \recordset$, we know that $\entropy(r) = 0$ for any quasi-entropy.
Therefore, $\entropy(I) = \log \numberin{I}$
is the unique quasi-entropy function for information states.  The 
quasi-entropy of a bit state is $\entropy(\bitstate) = 1$.

\section{Demons}  \label{sec:demons}

Maxwell's demon accomplishes changes in thermodynamic
states by acquiring and using information.  For example, a demon that
operates a trapdoor between two containers of gas can arrange for all
of the gas molecules to end up in one of the containers, ``compressing''
the gas without work.  As we have seen, it is also possible to imagine 
a reversible demon, which acquires and manipulates information 
in a completely reversible way.  If such a demon produces 
a transformation from state $x$ to state $y$
by acquiring $k$ bits of information in its memory, it can accomplish the reverse
transformation (from $y$ to $x$) while erasing $k$ bits from its memory.

Maxwell's demon is a key concept in axiomatic information thermodynamics,
and we introduce a new axiom to describe ``demonic'' processes.
\begin{axiom}Demons: \label{axiom:demons}
	Suppose $a,b \in \stateset$ and $J \in \infoset$ such that $a \rightarrow b + J$.
	\begin{description}
		\item(a)  There exists $I \in \infoset$ such that $b \rightarrow a+I$.
		\item(b)  For any $I \in \infoset$, either $a \rightarrow b + I$ or $b+I \rightarrow a$.
	\end{description} 
\end{axiom}
In Part (a), we assert that what one demon can do (transforming $a$ to $b$
by acquiring information in $J$), another demon can undo (transforming
$b$ to $a$ by acquiring information in $I$).  Part (b) envisions a reversible
demon.  Any amount of information in $I$ is either large enough that we can
turn $a$ to $b$ by acquiring $I$, or small enough that we can erase the information
by turning $b$ to $a$.

A process $\fproc{A}{B}$ is said to be \emph{demonically possible} if one
of two equivalent conditions hold:
\begin{itemize}
	\item  There exists an information state $J \in \infoset$ such that 
		either $A \rightarrow B + J$ or $B \rightarrow A + J$.
	\item  There exists an information process $\fproc{I}{J} \in \procset_{I}$ such
		that $\fproc{A}{B} + \fproc{I}{J}$ is possible; that is, either 
		$A + I \rightarrow B + J$ or $B + J \rightarrow A + I$.
\end{itemize}
It is not hard to see that these are equivalent.  Suppose we have
$J \in \infoset$ such that $A \rightarrow B + J$.  Then, for any 
$I \in \infoset$, $A + I \rightarrow B + (I+J)$.  Conversely, we 
note that $A \rightarrow A + I$ for any $I \in \infoset$.  Thus,
if $A + I \rightarrow B + J$ then $A \rightarrow B + J$ as
well.

If a process is possible, then it is also demonically possible, since if
$A \rightarrow B$ it is also true that $A + I \rightarrow B + I$.  For
singleton processes in $\procset_{S}$, moreover, the converse is
also true.  Suppose $a,b \in \stateset$ and $\fproc{a}{b}$ is demonically possible.
Then, there exists $I \in \infoset$ such that either $a \rightarrow b + I$ or
$b \rightarrow a + I$.  Either way, either trivially or by an application
of Axiom~\ref{axiom:demons}, there must be $J \in \infoset$ so that
$a \rightarrow b + J$.  A single record state $r \in \recordset$ is a
singleton information state in $\infoset$.  Thus, by Axiom~\ref{axiom:demons}
it must be that either $a \rightarrow b + r$ or $b + r \rightarrow a$.
Since $b + r \leftrightarrow b$, we find that $a \leftorright b$,
and so $\fproc{a}{b}$ is possible.

A singleton process is demonically possible if and only if it is
possible.  This means that we can use processes involving
demons to understand processes that do not.  In fact, we can
use Axiom~\ref{axiom:demons} to prove a highly significant
fact about the $\rightarrow$ relation on singleton eidostates.

\begin{theorem}  \label{theorem:comparison}
	Suppose $a,b,c \in \stateset$.  If $\fproc{a}{b}$
	and $\fproc{a}{c}$ are possible, then $\fproc{b}{c}$
	is possible.
\end{theorem}
\begin{proof}
	First, we note the general fact that, if $\fproc{x}{y}$ is a possible
	singleton process, then there exist $I,J \in \infoset$ so
	that $x \rightarrow y + I$ and $y \rightarrow x + J$.
	
	Given our hypothesis, therefore, there must be $I,J \in \infoset$
	such that $b \rightarrow a + I$ and $a \rightarrow c + J$.
	Then
	\begin{equation}
		b \rightarrow a + I \rightarrow c + (I + J).
	\end{equation}
	That is, $\fproc{b}{c}$ is demonically possible, and hence possible.
\end{proof}

This fact is so fundamental that Giles made it an axiom in his 
theory.  For us, it is a straightforward consequence of the
axiom about processes involving demons.  It tells us that 
the set of singleton states $\stateset$ is partitioned into 
equivalence classes, within each of which all states are
related by $\leftorright$.

This statement is more primitive than, but closely related to, 
a well-known principle called the Comparison Hypothesis.
The Comparison Hypothesis deals with a state relation called 
\emph{adiabatic accessibility} (denoted $\prec$) which is 
definable in Giles's theory (and ours) but is taken
as an undefined relation in some other axiomatic developments.
According to the Comparison Hypothesis, if $X$ and $Y$ are
states in a given thermodynamic space, either $X \prec Y$ or
$Y \prec X$.  Lieb and Yngvason, for instance, show that
the Comparison Hypothesis can emerge as a consequence of
certain axioms for thermodynamic states, spaces, and equilibrium
\cite{LiebYngvason1999,LiebYngvason1998}.

Theorem~\ref{theorem:comparison} also sheds light on our axiom 
about conditional processes, Axiom~\ref{axiom:conditional}.
In Part (a) of this axiom, we suppose that $A \rightarrow b$ for
some $A \in \eidoset$ and $b \in \stateset$.  The axiom itself
allows us to infer that $a \rightarrow b$ for every $a \in A$.
However, Theorem~\ref{theorem:comparison} now implies that,
for every $a,a' \in A$, either $a \rightarrow a'$ or $a' \rightarrow a$.
In other words, Part (a) of Axiom~\ref{axiom:conditional}, like
Part (b) of the same axiom, only applies to uniform eidostates.

Finally, we introduce one further ``demonic'' axiom.
\begin{axiom}Stability: \label{axiom:stability}
	Suppose $A,B \in \eidoset$ and $J \in \infoset$.  If $nA \rightarrow nB + J$ for
	arbitrarily large values of $n$, then $A \rightarrow B$.
\end{axiom}
According to the Stability Axiom, if a demon can transform arbitrarily
many copies of eidostate $A$ into arbitrarily many copies of $B$
while acquiring a bounded amount of information, then we may
say that $A \rightarrow B$.  This can be viewed as a kind of 
``asymptotic regularization'' of the $\rightarrow$ relation. The
form of the Stability Axiom that we have chosen is a particularly simple one,
and it suffices for our purposes in this paper.  However, 
more sophisticated axiomatic developments might require a 
refinement of the axiom.  Compare, for instance, Axiom 2.1.3
in Giles to its refinement in Axiom 7.2.1 \cite{Giles1964}.

To illustrate the use of the Stability Axiom, suppose that $A,B \in \eidoset$
and $I \in \infoset$ such that $A + I \rightarrow B + I$.  Our axioms do not
provide a ``cancellation law'' for information states, so we cannot immediately 
conclude that $A \rightarrow B$.  However, we can show that $nA \rightarrow nB + I$
for all positive integers $n$.  The case $n=1$ holds since $A \rightarrow A + I 
\rightarrow B + I$.  Now, we proceed inductively, assuming that 
$nA \rightarrow nB + I$ for some $n$.  Then,
\begin{eqnarray}
	(n+1) A & \rightarrow & nA + A \nonumber \\
		& \rightarrow & (nB + I) + A \nonumber \\
		& \rightarrow & nB + (A + I) \nonumber \\
		& \rightarrow & nB + (B + I) \rightarrow (n+1) B + I .
\end{eqnarray}
Thus, $nA \rightarrow nB + I$ for arbitrarily large (and indeed all) values of $n$.
By the Stability Axiom, we see that $A \rightarrow B$.  Thus, there is after all a 
general cancellation law information states that appear on both sides of 
the $\rightarrow$ relation.

\section{Irreversibility for Singleton Processes}  \label{sec:irreversibility}

From our two ``demonic'' axioms (Axioms~\ref{axiom:demons} and \ref{axiom:stability}),
we can use the properties of information states to derive an irreversibility
function on singleton processes.  Let $\hat{\procset}_{S}$ denote the set
of possible singleton processes.  This is a subgroup of the Abelian 
group $\procset_{S}$.  Thus, if we can find an irreversibility function $\irrev$
on $\hat{\procset}_{S}$, we will be able to extend it to all of $\procset_{S}$.

We begin by proving a useful fact about possible singleton processes:
\begin{theorem}  \label{theorem:alphabits}
	Suppose $a,b \in \stateset$ so that $\alpha = \fproc{a}{b}$ is possible.
	Then, for any integers $n,m \geq 0$, either $a + m \bitstate \rightarrow b + n \bitstate$
	or $b + n \bitstate \rightarrow a + m \bitstate$.
\end{theorem}
\begin{proof}
	If $m=n$, the result is easy.  Suppose that $n > m$, so that $n = m + k$ for 
	positive integer $k$.  By Axiom~\ref{axiom:demons}, either $a \rightarrow b+k \bitstate$
	or $b + k \bitstate \rightarrow a$.  We can then append the information state $m \bitstate$
	to both sides and rearrange the components.  The argument for $m > n$ is exactly
	similar.
\end{proof}
This fact has a corollary that we may state using the $\rightarrow$ relation on processes.
Suppose $\alpha = \fproc{a}{b}$ is a possible singleton process, and let $q,p$ be 
integers with $q > 0$.  Then, either $q \alpha \rightarrow p \bitproc$ 
or $q \alpha \leftarrow p \bitproc$ (so that $q \alpha - p \bitproc$ is either
natural or antinatural) for the bit process $\bitproc = \fproc{r}{\bitstate}$.

Given $\alpha$, therefore, we can define two sets of rational numbers:
\begin{equation}
	L_{\alpha} = \{ p/q : q\alpha \rightarrow p \bitproc \} \qquad
	U_{\alpha} = \{ p/q : q\alpha \leftarrow p \bitproc \}
\end{equation}
where $q > 0$.  Both sets are nonempty and every rational number is in
at least one of these sets.  Furthermore, if $p/q \in U_{\alpha}$ and
$p'/q' \in L_{\alpha}$, we have that $q \alpha \leftarrow p \bitproc$ and
$q' \alpha \rightarrow p' \bitproc$, and so

\begin{equation}
	pq' \bitproc \rightarrow qq' \alpha \rightarrow p'q \bitproc .
\end{equation}
Hence, $(pq' - p'q) \bitproc \rightarrow \zeroprocess$.  Since $\bitproc$
is itself a natural irreversible process, it follows that $pq' - p'q \geq 0$,
and so

\begin{equation}
	\frac{p}{q} \geq \frac{p'}{q'} .
\end{equation}
Every element of $U_{\alpha}$ is an upper bound for $L_{\alpha}$.
It follows that $U_{\alpha}$ and $L_{\alpha}$ form a Dedekind cut
of the rationals, which leads us to the following important result.

\begin{theorem}  \label{theorem:singletonirreversibility}
For $\alpha \in \hat{\procset}_{S}$, define $\irrev(\alpha) = \inf U_{\alpha}
= \sup L_{\alpha}$.  
Then, $\irrev$ is an irreversibility function on $\hat{\procset}_{S}$.
\end{theorem}
\begin{proof}
	First, we must show that $\irrev$ is additive.  
	Suppose $\alpha, \beta \in \hat{\procset}_{S}$.
	If $p/q \in U_{\alpha}$ and $p'/q \in U_{\beta}$,
	then $q (\alpha + \beta) \rightarrow (p + p') \bitproc$,
	and so $(p+p')/q \in U_{\alpha + \beta}$.  	It follows
	that $\irrev(\alpha + \beta) \leq \irrev(\alpha) + \irrev(\beta)$.
	The corresponding argument involving $L_{\alpha}$ and 
	$L_{\beta}$ proves that 
	$\irrev(\alpha + \beta) \geq \irrev(\alpha) + \irrev(\beta)$.
	Thus, $\irrev$ must be additive.
	
	Next, we must show that the value of $\irrev(\alpha)$
	tells us the type of the process $\alpha$.
	If $\irrev(\alpha) > 0$, then $0 \in L_{\alpha}$ but 
	$0 \notin U_{\alpha}$, and so $\alpha \rightarrow \zeroprocess$
	but $\alpha \nleftarrow \zeroprocess$.  That is, $\alpha$
	is natural irreversible.  Likewise, if $\irrev(\alpha) < 0$,
	then $0 \notin L_{\alpha}$ but $0 \in U_{\alpha}$, from which
	we find that $\alpha$ must be antinatural irreversible.
	Finally, if $\irrev(\alpha) = 0$, we find that 
	$q \alpha - \bitproc \rightarrow \zeroprocess$ and
	$q \alpha + \bitproc \leftarrow \zeroprocess$ for
	arbitrarily large values of $q$.  From Axiom~\ref{axiom:stability},
	we may conclude that $\alpha \leftrightarrow \zeroprocess$,
	and so $\alpha$ is reversible.
\end{proof}

Notice that we have arrived at an irreversibility function for possible
singleton processes---those most analogous to the ordinary
processes in Giles or any text on classical thermodynamics---from axioms
about information and processes involving demons
(Axioms~\ref{axiom:information}--\ref{axiom:stability}).
In our view, such ideas are not ``extras'' to be appended 
onto a thermodynamic theory, but are instead central
concepts throughout.  In ordinary thermodynamics, the
possibility of a reversible heat engine can
have implications for processes that do not involve
any heat engines at all.  In the information thermodynamics 
whose axiomatic foundations we are exploring, 
the possibility of a Maxwell's demon has implications
even for situations in which no demon acts.

We now have irreversibility functions for both information
processes and singleton processes.  These are closely
related.  In fact, it is possible to prove the following general
result:  
\begin{theorem}  \label{theorem:irreversibilityonsi}
	If $\alpha \in \hat{\procset}_{S}$ and $\Gamma \in \procset_{I}$,
	then the combined process $\alpha + \Gamma$ is natural
	if and only if $\irrev(\alpha) + \irrev(\Gamma) \geq 0$.
\end{theorem}

Since the set $\hat{\procset}_{S}$ of possible singleton processes 
is a subgroup of the Abelian group $\procset_{S}$ of all singleton
processes, we can extend the additive irreversiblity function 
$\irrev$ to all of $\procset_{S}$.  Though $\irrev$ is unique on
the possible set $\hat{\procset}_{S}$, its extension to $\procset_{S}$
is generally not unique.

\section{Components of Content and Entropy}   \label{sec:mechanical}

Our axiomatic theory of information thermodynamics is fundamentally
about the set of eidostates $\eidoset$.  However, the part 
of that theory dealing with the set $\stateset$ of singleton 
eidostates includes many of the concepts and results 
of ordinary axiomatic thermodynamics \cite{Giles1964}.  We have a group
of singleton processes $\procset_{S}$ containing a subgroup
$\hat{\procset}_{S}$ of possible processes, and we have
constructed an irreversibility function $\irrev$ on $\hat{\procset}_{S}$
that may be extended to all of $\procset_{S}$.  From these
we can establish several facts.
\begin{itemize}
	\item  We can construct \emph{components of content}, which
		are the abstract versions of conserved quantities.
		A component of content $Q$ is an additive function
		on $\stateset$ such that, if the singleton process
		$\fproc{a}{b}$ is possible, then $Q(a) = Q(b)$.
		(In conventional thermodynamics, components of
		content include energy, particle number, etc.)
	\item  We can find a \emph{sufficient set} of components
		of content.  The singleton process $\fproc{a}{b}$
		is possible if and only if $Q(a) = Q(b)$ for all
		$Q$ in the sufficient set.
	\item  We can use $\irrev$ to define a \emph{quasi-entropy}
		$\entropy$ on $\stateset$ as follows:
		$\entropy(a) = \irrev(\fproc{a}{2a})$.
		This is an additive function on states in $\stateset$
		such that $\irrev(\fproc{a}{b}) = \entropy(b) - \entropy(a)$.
\end{itemize}
Because the extension of the irreversibility function $\irrev$ from
$\hat{\procset}_{S}$ to all of $\procset_{S}$ is not unique, the
quasi-entropy $\entropy$ is not unique either.  How could 
various quasi-entropies differ?  Suppose 
$\irrev_{1}$ and $\irrev_{2}$ are two different extensions
of the same original $\irrev$, leading to two quasi-entropy
functions $\entropy_{1}$ and $\entropy_{2}$ on $\stateset$.
Then, the difference $Q = \entropy_{1} - \entropy_{2}$ 
is a component of content.  That is, if $\fproc{a}{b}$ is possible,
\begin{eqnarray}
	Q(b) - Q(a) 
	& = & \entropy_{1}(b) - \entropy_{2}(b) - \entropy_{1}(a) + \entropy_{2}(a) \nonumber \\
	& = & \irrev_{1}(\fproc{a}{b}) - \irrev_{2}(\fproc{a}{b}) \nonumber \\
	& = & 0,
\end{eqnarray}
since $\irrev_{1}$ and $\irrev_{2}$ agree on $\hat{\procset}_{S}$,
which contains $\fproc{a}{b}$.

Another idea that we can inherit without alteration is the concept 
of a \emph{mechanical state}.  A mechanical state is a singleton
state that reversibly stores one or more components of content,
in much the same way that we can store energy reversibly
as the work done to lift or lower a weight.  The mechanical state
in this example is the height of the weight.  In Giles's theory \cite{Giles1964},
mechanical states are the subject of an axiom, which we
also adopt:
\begin{axiom}Mechanical states: \label{axiom:mechanical}
	There exists a subset $\mechset \subseteq \stateset$ of 
	\emph{mechanical states} such that:
	\begin{description}
		\item(a)  If $l,m \in \mechset$, then $l+m \in \mechset$.
		\item(a)  For $l,m \in \mechset$, if $l \rightarrow m$
			then $m \rightarrow l$.
	\end{description}	
\end{axiom}
Nothing in this axiom asserts the actual \emph{existence} of
any mechanical state.  It might be that $\mechset = \emptyset$.
Furthermore, the choice of the designated set $\mechset$ 
is not determined solely by the $\rightarrow$ relations among the 
states.  For instance, the set $\recordset$ of record states might 
be included in $\mechset$, or not.  
This explains why the introduction of mechanical states must be
phrased as an axiom, rather than a definition: a complete
specification of the system must include the choice of which
set is to be designated as $\mechset$.
Whatever choice is made for 
$\mechset$, the set $\procset_{M}$ of \emph{mechanical processes} 
(i.e., those equivalent to $\fproc{l}{m}$ for $l,m \in \mechset$) 
will form a subgroup of $\procset_{S}$.

A mechanical state may ``reversibly store'' a component of content $Q$,
but it need not be true that every $Q$ can be stored like this.  We
say that a component of content $Q$ is \emph{non-mechanical} if
$Q(m) = 0$ for all $m \in \mechset$.  For example, we might store
energy by lifting or lowering a weight, but we cannot store particle
number in this way.

Once we have mechanical states and processes, we can give
a new classification of processes.  A process $\Gamma \in \procset$ 
is said to be \emph{adiabatically natural} (\emph{possible}, \emph{reversible},
\emph{antinatural}) if there exists a mechanical process
$\mu \in \procset_{M}$ such that $\Gamma + \mu$ is
natural (possible, reversible, antinatural).  We can also define
the ``adiabatic accessibility'' relation for states in $\stateset$, as
mentioned in Section~\ref{sec:demons}:  $a \prec b$ whenever
the process $\fproc{a}{b}$ is adiabatically natural.

The set $\mechset$ of mechanical states allows us to 
refine the idea of a quasi-entropy into an \emph{entropy},
which is a quasi-entropy $\entropy$ that takes the
value $\entropy(m) = 0$ for any mechanical state $m$.
Such a function is guaranteed to exist.  We end up with
a characterization theorem, identical to a result of Giles \cite{Giles1964},
that summarizes the general thermodynamics of singleton 
eidostates in our axiomatic theory.
\begin{theorem} \label{theorem:singletonthermodynamics}
	There exist an entropy function $\entropy$ and a set of 
	components of content $Q$ on $\stateset$ with the following~properties:
	\begin{description}
		\item(a)  For any $a,b \in \stateset$, $\entropy(a+b) = \entropy(a) + \entropy(b)$.
		\item(b)  For any $a,b \in \stateset$ and component of content $Q$, 
			$Q(a+b) = Q(a)+Q(b)$.
		\item(c)  For any $a,b \in \stateset$, $a \rightarrow b$ if and only if
			$\entropy(a) \leq \entropy(b)$ and $Q(a) = Q(b)$ for every
			component of content $Q$.
		\item(d)  $\entropy(m) = 0$ for all $m \in \mechset$.
	\end{description}
\end{theorem}
An entropy function is not unique.  Two entropy functions may differ
by a non-mechanical component of content.

The entropy function on $\stateset$ is related to the information
entropy function we found for information states in $\infoset$.  Suppose
we have $a,b \in \stateset$ and $I,J \in \infoset$.  Then,
Theorem~\ref{theorem:irreversibilityonsi} tells us that
$a + I \rightarrow b + J$ only~if 

\begin{equation}
	\entropy(a) + \log \numberin{I} \leq \entropy(b) + \log \numberin{J}.
\end{equation}

Let $\eidoset_{SI}$ represent the set of eidostates that are
similar to a singleton state combined with an information state.
Then, $\entropy(a + I) = \entropy(a) + \log \numberin{I}$ is an
entropy function on $\eidoset_{SI}$.  In the next section, we will
extend the domain of the entropy function even further, to the set
$\uniformset$ of all uniform eidostates.

\section{State Equivalence}  \label{sec:equivalence}

Consider a thought-experiment (illustrated in Figure~\ref{fig:macroaxiom}) 
in which a one-particle gas
starts out in a volume $v_{0}$ and a second thermodynamic
system starts out in one of three states $e_1$, $e_2$ or $e_3$.
We assume that all conserved quantities are the same for
these three states, but they may differ in entropy.  We can formally
describe the overall situation by the eidostate $E + v_{0}$,
where $E = \{ e_{1}, e_{2}, e_{3} \}$ is uniform.

\begin{figure}
\begin{center}
\includegraphics[width=4.5in]{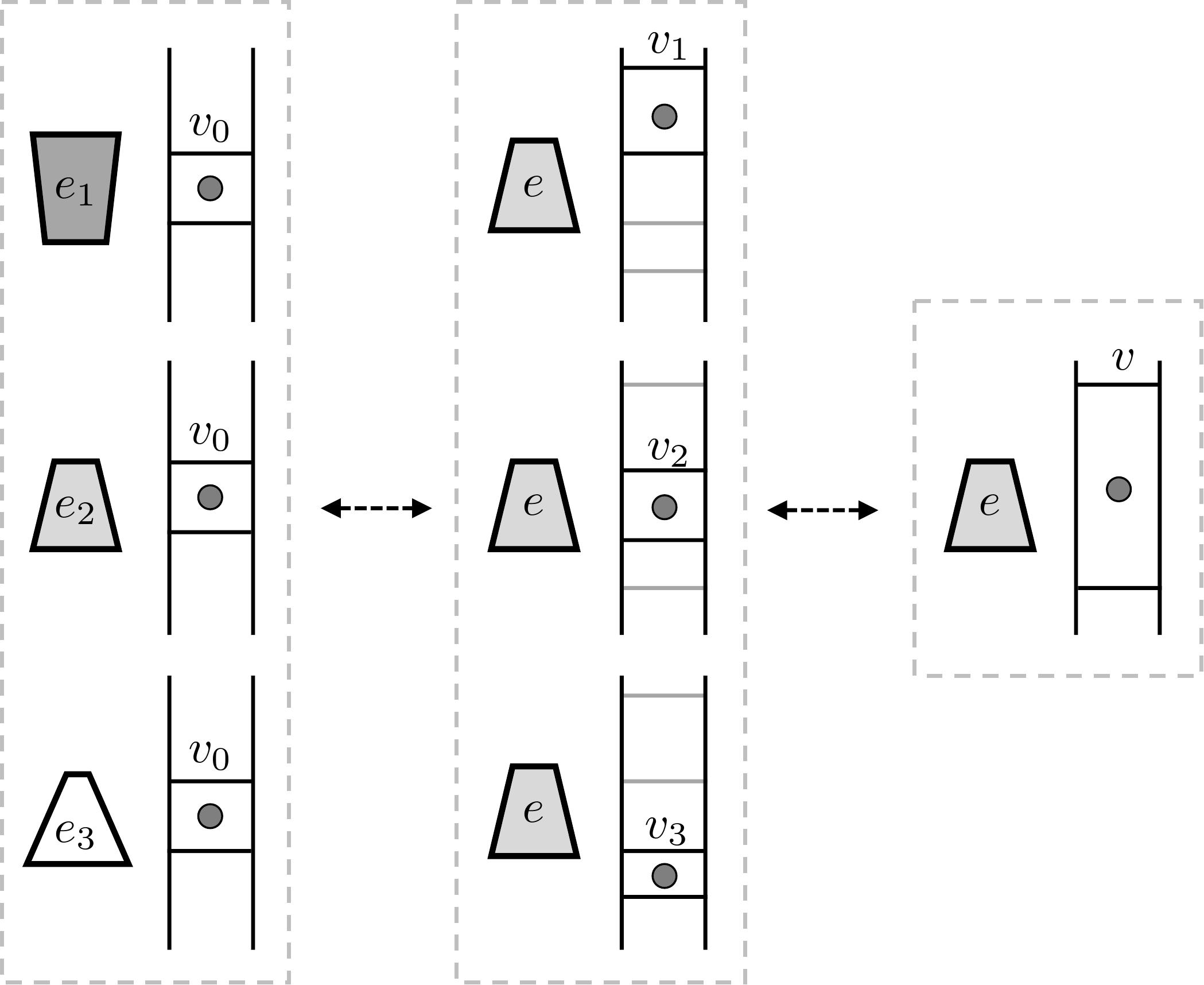}
	\end{center}
	\caption{\label{fig:macroaxiom} A state equivalence
		thought-experiment involving states of an 
		arbitrary thermodynamic system and a one-particle gas.}
\end{figure}

We can reversibly transform each of the $e_{k}$ states to
the same state $e$, compensating for the various changes
in entropy by expanding or contracting the volume occupied
by the gas.  That is, we can have $v_{0} + e_{k} \leftrightarrow v_{k} + e$
for adjacent but non-overlapping volumes $v_{k}$.  
Axiom~\ref{axiom:conditional} indicates that we can write
$E + v_{0} \leftrightarrow e + V$, where $V = \{ v_{1}, v_{2}, v_{3} \}$.

Now, we note that $V$ can itself be reversibly transformed into
a singleton eidostate $v$.  A gas molecule in one of the
sub-volumes (eidostate $V$) can be turned into a molecule
in the whole volume (eidostate $e$) by removing internal
partitions between the sub-volumes; and when we re-insert
these partitions the particle is once again in just one of 
sub-volumes.  Thus, $V \leftrightarrow v$.  To summarize,
we have 

\begin{equation}
	E + v_{0} \leftrightarrow e + v .
\end{equation}

The uniform eidostate $E$, taken together with the gas state $v_{0}$,
can be reversibly transformed into the singleton state $v + e$.
We call this a \emph{state equivalence} for the uniform eidostate $E$.
By choosing the state $e$ properly (say, by letting $e = e_{k}$ 
for some $k$), we can also guarantee that the volume $v$ is 
larger than $v_{0}$, so that $v_{0} \rightarrow v$ by free
expansion.

This discussion motivates our final axiom, which states that 
this type of reversible transformation is always possible for
a uniform eidostate.
\begin{axiom}State equivalence:  \label{axiom:equivalence}
	If $E$ is a uniform eidostate then there exist states $e,x,y \in \stateset$
	such that $x \rightarrow y$ and $E + x \leftrightarrow e + y$.
\end{axiom}

Axiom~\ref{axiom:equivalence} closes a number of gaps in our theory.
For example, our previous axioms (Axioms \ref{axiom:eidostates}--\ref{axiom:mechanical}) 
do not by themselves guarantee than \emph{any} state in $\stateset$ has 
a nonzero entropy.  With the new axiom, 
however, we can prove that such states exist. 
The bit state $\bitstate$, which is uniform, has some state equivalence
given by $\bitstate + x \leftrightarrow e + y$.  Since the eidostates
on each side of this relation are in $\eidoset_{SI}$, we can determine
the entropies on each side.  We find that

\begin{equation}
	1 + \entropy(x) = \entropy(e) + \entropy(y) .
\end{equation}
It follows that at least one of the states $e,x,y$ must have 
$\entropy \neq 0$.

State equivalence also allows us to define the entropy of any uniform
eidostate $E \in \uniformset$.  If $E + x \leftrightarrow e + y$ then we let

\begin{equation}
	\entropy(E) = \entropy(e) + \entropy(y) - \entropy(x) .
\end{equation}
We must first establish that this expression is well-defined.
If we have two state equivalences for the same $E$,
so that $E + x \leftrightarrow e + y$ and $E + x' \leftrightarrow e' + y'$,
then

\begin{equation}
	(e + y) + x' \leftrightarrow E + (x + x') \leftrightarrow (e' + y') + x,
\end{equation}
from which it follows that

\begin{equation}
	\entropy(e) + \entropy(y) - \entropy(x) = \entropy(e') + \entropy(y') - \entropy(x').
\end{equation}
Thus, our definition for $\entropy(E)$ does not depend on our
choice of state equivalence for $E$.

Is $\entropy$ an entropy function on $\uniformset$?
It is straightforward to show that $\entropy$ is additive on $\uniformset$
and that $\entropy(m) = 0$ for any mechanical state $m$.  It remains
to show that, for any $E,F \in \uniformset$ with $\fproc{E}{F}$ possible,
$E \rightarrow F$ if and only if $\entropy(E) \leq \entropy(F)$.  We will
use the state equivalences $E + x \leftrightarrow e + y$ and 
$F + w \leftrightarrow f + z$.

Suppose first that $E \rightarrow F$.  Then, $E + (x+w) \rightarrow F + (x+w)$
and so

\begin{equation}
	(e+y) + w \leftrightarrow (E+x) + w \rightarrow (F+w) + x \leftrightarrow (f+z) + x.
\end{equation}
From this, it follows that

\begin{equation}
	\entropy(e) + \entropy(y) - \entropy(x) \leq \entropy(f) + \entropy(z) - \entropy(w),
\end{equation}
and hence $\entropy(E) \leq \entropy(F)$.

We can actually extract one more fact from this argument.  If we assume
that $\fproc{E}{F}$ is possible, it must also be true that
$\fproc{(e+y)+w}{(f+z)+x}$ is a possible singleton process.
If we now suppose that $\entropy(E) \leq \entropy(F)$, we know that
$\entropy((e+y)+w) \leq \entropy((f+z)+x)$ and thus

\begin{equation}
	E+(x+w) \leftrightarrow (e+y)+w \rightarrow (f+z)+x \leftrightarrow F + (x+w).
\end{equation}
Therefore, $E \rightarrow F$, as desired.

We have extended the entropy $\entropy$ to uniform eidostates.
It is even easier to extend any component of content function $Q$
to these states.  If $E \in \uniformset$, then any $e_{1},e_{2} \in E$
must have $Q(e_{1}) = Q(e_{2})$, since $e_{1} \leftorright e_{2}$.  
Thus, we can define $Q(E) = Q(e_{k})$ for any $e_{k} \in E$.
This is additive because the elements of $E+F$ are combinations
$e_{k}+f_{j}$ of states in $E$ and $F$.
Furthermore, suppose we have a state equivalence $E + x \leftrightarrow e + y$.
Since we assume $x \rightarrow y$ in a state equivalence, $Q(x) = Q(y)$.  
By the conditional process axiom
(Axiom~\ref{axiom:conditional}) we know that $e_{k} + x \rightarrow e + y$
for any $e_{k} \in E$.  It follows that $Q(e) = Q(e_{k}) = Q(E)$.

Now, let $E_{1}$ and $E_{2}$ be uniform eidostates with
state equivalences $E_{k} + x_{k} \leftrightarrow e_{k} + y_{k}$.
Suppose further that $Q(E_{1}) = Q(E_{2})$ for every
component of content $Q$.  We know that 
$Q(x_{1}) = Q(y_{1})$, $Q(x_{2}) = Q(y_{2})$ and
$Q(e_{1}) = Q(e_{2})$ for every component of content.
Thus
\begin{equation}
	E_{1} + (x_{1} + x_{2}) \leftrightarrow (e_{1} + y_{1}) + x_{2}
		\leftorright (e_{2} + y_{2}) + x_{1} \leftrightarrow
		E_{2} + (x_{1} + x_{2}) .
\end{equation}
It follows that $E_{1} \leftorright E_{2}$, i.e., that $\fproc{E_{1}}{E_{2}}$
is a possible eidostate process.

We have therefore extended Theorem~\ref{theorem:singletonthermodynamics}
to all uniform eidostates.  We state the new result here.
\begin{theorem}[Uniform eidostate thermodynamics] \label{theorem:uniformthermodynamics}
	There exist an entropy function $\entropy$ and a set of 
	components of content $Q$ on $\uniformset$ with the following properties:
	\begin{description}
		\item(a)  For any $E,F \in \uniformset$, $\entropy(E+F) = \entropy(E) + \entropy(F)$.
		\item(b)  For any $E,F \in \uniformset$ and component of content $Q$, 
			$Q(E+F) = Q(E)+Q(F)$.
		\item(c)  For any $E,F \in \uniformset$, $E \rightarrow F$ if and only if
			$\entropy(E) \leq \entropy(F)$ and $Q(E) = Q(F)$ for every
			component of content~$Q$.
		\item(d)  $\entropy(m) = 0$ for all $m \in \mechset$.
	\end{description}
\end{theorem}
The set $\uniformset$ of uniform eidostates includes the singleton
states in $\stateset$, the information states in $\infoset$, all
combinations of these, and perhaps many other states as well.
(Non-uniform eidostates in $\eidoset$ \emph{might} exist, as we will see
in the model we discuss in Section~\ref{sec:model1},
but their existence cannot be proved from our axioms.)
The type of every process involving uniform eidostates can be determined
by a single entropy function (which must not decrease) and a 
set of components of content (which must be conserved).

\section{Entropy for Uniform Eidostates}  \label{sec:uniformentropy}

We have extended the entropy function from singleton states and
information states to all uniform eidostates.  It turns out that this 
extension is unique.  The following theorem and its corollaries
actually allow us to compute the entropy of any $E \in \uniformset$
from the entropies of the states contained in $E$.
\begin{theorem} \label{theorem:decomposeentropy}
	Suppose $E \in \uniformset$ is a disjoint union of uniform eidostates $E_{1}$ and
	$E_{2}$.  Then
	
	\begin{equation}
		\entropy(E) = \entropy(E_{1} \cup E_{2}) = \log \left ( 2^{\entropy(E_{1})} + 2^{\entropy(E_{2})} \right ) .
	\end{equation}
\end{theorem}
\begin{proof}
	$E$ and $E_{k}$ ($k=1,2$) all have equal components of content.
	Define
	
	\begin{equation}
		\Delta(E_{1},E_{2}) 
		= \entropy(E) - \log \left ( 2^{\entropy(E_{1})} + 2^{\entropy(E_{2})} \right ) .
	\end{equation}
	Note that, if we replace $E_{k}$ by $E_{k}' = E_{k} + I$ for $I \in \infoset$, then
	these new eidostates are still disjoint and
	
	\begin{equation}
		E' = E_{1}' \cup E_{2}' = 
		\left ( E_{1} + I \right ) \cup \left ( E_{2} + I \right ) = \left ( E + I \right ) .
	\end{equation}
	These eidostates have the same components of content as the original $E$.  Furthermore,
	\begin{eqnarray}
		\Delta(E_{1}',E_{2}') 
		& = & \Delta ( E_{1}+I , E_{2}+I ) \nonumber \\
		& = & \entropy(E+I) - \log \left ( 2^{\entropy(E_{1}+I)} + 2^{\entropy(E_{2}+I)} \right ) 
			\nonumber \\
		& = & \entropy(E) + \log \numberin{I} 
			- \log \left ( 2^{\log \numberin{I}} \left ( 2^{\entropy(E_{1})} + 2^{\entropy(E_{2})} \right ) 
			\right ) \nonumber \\
		& = & \Delta (E_{1},E_{2}) .
	\end{eqnarray}

	We can find a uniform eidostate $E_{0}$ with the same components of content such that 
	$S_{0} = \entropy(E_{0})$ is less than
	or equal to $\entropy(E_{1})$, $\entropy(E_{2})$ and $\entropy(E)$.  (It suffices to pick
	$E_{0}$ to be the state of smallest entropy among $E_{1}$, $E_{2}$ and $E$.)  Then, there exist
	integers $m_{k} \geq 1$ such that
	\begin{itemize}
		\item  There are disjoint information states $J_{k}$ containing
			$m_{k}$ record states.
		\item  There are disjoint information states $J_{k}^{\ast}$ containing
			$m_{k} + 1$ record states.
		\item  $J = J_{1} \cup J_{2}$ and $J^{\ast} = J_{1}^{\ast} \cup J_{2}^{\ast}$ have
			$m_{1}+m_{2}$ and $m_{1}+m_{2}+2$ record states, respectively.
		\item  We have
			\begin{equation}
				S_{0} + \log m_{k} \leq \entropy(E_{k}) < S_{0} + \log (m_{k}+1) .
			\end{equation}
	\end{itemize}
	That is, we choose $m_{k}$ so that $\entropy(E_{k}) - S_{0} \geq 0$ is
	between $\log(m_{k})$ and $\log(m_{k}+1)$.  To put it more simply,
	$m_{k} = \lfloor 2^{\entropy(E_{k}) - S_{0}} \rfloor$.
	
	Once we have $m_{k}$, we can write that
	
	\begin{equation}
		2^{S_{0}} \cdot m_{k} \leq 2^{\entropy(E_{k})} < 2^{S_{0}} \cdot (m_{k} + 1).
	\end{equation}
	Adding these inequalities for $k=1,2$ and taking the logarithm yields
	
	\begin{equation}
		S_{0} + \log (m_{1}+m_{2}) 
			\leq \log \left ( 2^{\entropy(E_{1})} + 2^{\entropy(E_{2})} \right )
			< S_{0} + \log (m_{1} + m_{2}+2) .
	\end{equation}	
	How far apart are the two ends of this chain of inequalities?  Here, is a 
	useful fact about base-2 logarithms:  If $n \geq 1$, 
	then $\log (n+2) < \log(n) + 2/n$.  This implies
	
	\begin{equation}
		S_{0} + \log (m_{1}+m_{2}) 
			\leq \log \left ( 2^{\entropy(E_{1})} + 2^{\entropy(E_{2})} \right )
			< S_{0} + \log (m_{1} + m_{2}) + \frac{2}{m_{1}+m_{2}} . 
	\end{equation}
	The two ends of the inequality differ by less than $2/(m_{1}+m_{2})$.
	
	We can get another chain of inequalities by applying Axiom~\ref{axiom:conditional}
	about conditional processes.  Since all of our uniform eidostates have the same
	components of content, we know that 
	
	\begin{equation}
		E_{0} + J_{k} \rightarrow E_{k} \rightarrow E_{0} + J_{k}^{\ast} .
	\end{equation}
	From the axiom, we can therefore say
	
	\begin{equation}
		E_{0} + J \rightarrow E \rightarrow E_{0} + J^{\ast} 
	\end{equation}
	which implies that
	\begin{eqnarray}
		S_{0} + \log(m_{1}+m_{2}) \leq \entropy(E) & \leq & S_{0} + \log(m_{1}+m_{2}+2) \nonumber \\
			& < & S_{0} + \log(m_{1}+m_{2}) + \frac{2}{m_{1}+m_{2}} .
	\end{eqnarray}
	We have two quantities that lie in the same interval.  Their separation is
	therefore bounded by the interval width---i.e., less than
	$2/(m_{1}+m_{2})$.  Therefore,
	\begin{eqnarray}
		\absolute{\Delta(E_{1},E_{2})} 
		& = & \absolute{\entropy(E) - \log \left ( 2^{\entropy(E_{1})} + 2^{\entropy(E_{2})} \right )}
			\nonumber \\
		& < & \frac{2}{m_{1}+m_{2}} .
	\end{eqnarray}
	
	How big are the numbers $m_{k}$?  We can make such numbers as large as we like
	by considering instead the eidostates $E_{k}' = E_{k} + I$, where $I$ is an information state.
	As we have seen, $\Delta(E_{1}',E_{2}')=\Delta(E_{1},E_{2})$.
	Given any $\epsilon > 0$, we can choose $I$ so that
	
	\begin{equation}
		m_{k}' = \left \lfloor  2^{\entropy(E_{k}') - S_{0}} \right \rfloor 
			= \left \lfloor  \numberin{I} \cdot 2^{\entropy(E_{k}) - S_{0}} \right \rfloor > \frac{1}{\epsilon} .
	\end{equation}
	Then, $2/(m_{1}'+m_{2}') < \epsilon$, and so
	
	\begin{equation}
		\absolute{\Delta(E_{1},E_{2})} = \absolute{\Delta(E_{1}',E_{2}')} < \epsilon .
	\end{equation}
	Since this is true for any $\epsilon > 0$, we must have $\Delta(E_{1},E_{2}) = 0$, and so
	
	\begin{equation}
		\entropy(E) = \log \left ( 2^{\entropy(E_{1})} + 2^{\entropy(E_{2})} \right ),
	\end{equation}
	as desired.
\end{proof}

Theorem~\ref{theorem:decomposeentropy} has a corollary, which we 
obtain by applying the theorem inductively:
\begin{theorem}  \label{theorem:entropyformula}
If $E$ is a uniform eidostate,

\begin{equation}
	\entropy(E) = \log \left (  \sum_{e_{k} \in E} 2^{\entropy(e_{k})}  \right ).
\end{equation}
\end{theorem}
The entropy of any uniform eidostate is a straightforward function of the entropies
of the states contained therein.  If $E$ contains more than one state,
we notice that $\entropy(E) > \entropy(e_{k})$ for any $e_{k} \in E$.
It follows that $\fproc{e_{k}}{E}$ is a natural irreversible process.

To take a simple example of Theorem~\ref{theorem:entropyformula}, 
consider the entropy of an information state.  
Every record state $r$ has $\entropy(r) = 0$.
Thus, for $I \in \infoset$,
\begin{equation}
	\entropy(I) = \log \left (  \sum_{r_{k} \in I} 2^{0}  \right ) = \log \numberin{I} ,
\end{equation}
as we have already seen.

\section{Probability}  \label{sec:probability}

An eidostate in $\eidoset$ represents a state of knowledge of 
a thermodynamic agent.  It is, as we have said, a simple
list of possible states, without any assignment of probabilities to them.
If the agent is to use probabilistic reasoning, then it needs
to assign conditional probabilities of the form $P(A|B)$ where
$A,B \in \eidoset$.

The entropy formula in Theorem~\ref{theorem:entropyformula}
allows us to make such an assignment based on the entropy itself,
provided the eidostate conditioned upon is uniform.
Suppose $E \in \uniformset$ and let $a \in \stateset$.  Then,
the \emph{entropic probability} of $a$ conditioned on $E$ is

\begin{equation}  \label{eq:entropicprob}
	P(a|E) = \left \{  \begin{array}{cl} 
		\displaystyle \frac{2^{\entropy(a)}}{2^{\entropy(E)}} = 2^{\entropy(a)-\entropy(E)}
			& \quad a \in E \\[2ex]
		0 & \quad a \notin E
	\end{array} \right .  .
\end{equation}
Clearly, $0 \leq P(a|E) \leq 1$ and $\displaystyle \sum_{a} P(a|E) = 1$.
It is worth noting that, although $E$ is ``uniform'' (in the sense that
all $a \in E$ have exactly the same conserved components of content),
the probability distribution $P(a|E)$ is {\em not} uniform, but assigns
a higher probability to states of higher entropy.

We can generalize entropic probabilities a bit further.
Let $A$ be any subset of $\stateset$ (be it an eidostate or not)
and $E \in \uniformset$.  Then, $A \cap E$
is either a uniform eidostate or the empty set $\emptyset$.  
If we formally assign $\entropy(\emptyset) = - \infty$, then
both $E$ and $A \cap E$ have well-defined entropies.
Then, we define

\begin{equation}
	p(A|E) = \sum_{a \in A} P(a|E) = 
	\frac{2^{\entropy(A \cap E)}}{2^{\entropy(E)}} 
		= 2^{\entropy(A \cap E) - \entropy(E)}.
\end{equation}
Obviously, $P(E|E) = 1$.  Now, consider two disjoint sets
$A$ and $B$ along with $E \in \uniformset$.  
The set $(A \cup B) \cap E$ is a disjoint union 
of uniform eidostates (or empty sets) $A \cap E$
and $B \cap E$.  Thus,
\begin{eqnarray}
	P(A \cup B | E) 
	& = & \frac{2^{\entropy((A \cup B) \cap E)}}{2^{\entropy(E)}} \nonumber \\
	& = & \frac{2^{\entropy(A \cap E)} + 2^{\entropy(B \cap E)}}{2^{\entropy(E)}} \nonumber \\[1ex]
	& = & P(A|E) + P(B|E) ,
\end{eqnarray}
in accordance with the rules of probability.
We also have the usual rule for conditional probabilities.
Suppose $A,B \subseteq \stateset$ and $E \in \uniformset$
such that $A \cap E \neq \emptyset$.  Then

\begin{equation}
	P(B|A \cap E) = \frac{2^{\entropy(B \cap (A \cap E))}}{2^{\entropy(A \cap E)}}
		= \frac{P(B \cap A | E)}{P(A | E)} .
\end{equation}

The entropy formula in Theorem~\ref{theorem:entropyformula} and the 
probability assignment in Equation~\ref{eq:entropicprob} call to mind familiar
ideas from statistical mechanics.  According to Boltzmann's formula, the
entropy is $\entropy = \log \Omega$, where $\Omega$ is (depending on
the context) the number (or phase space volume or Hilbert space dimension) of the
microstates consistent with macroscopic data about a system.  In the
microcanonical ensemble, a uniform probability distribution is assigned to
these microstates.  In an eidostate $E$ comprising non-overlapping  macrostates
states $\{ e_{1}, e_{2}, \ldots \}$, we would therefore expect 
$\Omega(E) = \Omega(e_{1}) + \Omega(e_{2}) + \ldots$, and the 
probability of state $e_{k}$ should be proportional to $\Omega(e_{k})$.
However, in our axiomatic system Theorem~\ref{theorem:entropyformula}
and Equation~(\ref{eq:entropicprob}) do not arise from any assumptions
about microstates, but solely from the ``phenomenological'' 
$\rightarrow$ relation among eidostates in $\eidoset$.

Even though the entropy function $\entropy$ is not unique, the entropic
probability assignment is unique.  Suppose $\entropy_{1}$ and $\entropy_{2}$
are two entropy functions for the same states and processes.  Then,
as we have seen, the difference $\entropy_{1} - \entropy_{2}$ is a 
component of content.  All of the states within a uniform eidostate $E$
(as well as $E$ itself) have the same values for all components of
content.  That is, $\entropy_{1}(a) - \entropy_{2}(a) = \entropy_{1}(E) - \entropy_{2}(E)$
for any $a \in E$.  Thus,

\begin{equation}
	P_{1}(a|E) = 2^{\entropy_{1}(a)-\entropy_{1}(E)} 
		= 2^{\entropy_{2}(a)-\entropy_{2}(E)} = P_{2}(a|E) .
\end{equation}
(Both probabilities are zero for $a \notin E$, of course.)

The entropic probability assignment is not the only possible
probability assignment, but it does have a number of 
remarkable properties.  For instance, 
suppose $E$ and $F$ are two uniform eidostates.  If we 
prepare them independently, we have the combined eidostate $E + F$
and the probability of some particular state $x+y$ is

\begin{equation}
	P(x+y | E+F) = \frac{2^{\entropy(x)+\entropy(y)}}{2^{\entropy(E) + \entropy(F)}}
		= P(x|E) \, P(y|F) ,
\end{equation}
as we would expect for independent events.

The entropic probability also yields an elegant expression for the
entropy $\entropy(E)$ of a uniform eidostate $E$.  
\begin{theorem}  \label{theorem:shannonentropy}
	Suppose $E$ is a uniform eidostate, and $P(a|E)$ is 
	the entropic probability for state $a \in E$.  Then
	
	\begin{equation}
		\entropy(E) = \biggl \langle \entropy(a) \biggr \rangle + H(\vec{P}),
	\end{equation}
	where $\langle \entropy(a) \rangle$ is the average state entropy in $E$
	and $H(\vec{P})$ is the Shannon entropy of the $P(a|E)$ distribution.
\end{theorem}
\begin{proof}
	We first note that, for any $a \in E$, 
	$\log P(a|E) = \entropy(a) - \entropy(E)$.
	We rewrite this as $\entropy(E) = \entropy(a) - \log P(a|E)$
	and take the mean value with respect to the $P(a|E)$
	probabilities:
	\begin{eqnarray}
		\entropy(E) 
		& = & \left \langle \entropy(E) \right \rangle \nonumber \\[1ex]
		& = & \sum_{a \in E} P(a|E) \left ( \entropy(a) - \log P(a|E) \right ) \nonumber \\
		& = & \sum_{a \in E} P(a|E) \entropy (a) - \sum_{a \in E} P(a|E) \log P(a|E),
	\end{eqnarray}
	Therefore, $\entropy(E) = \left \langle \entropy(a) \right \rangle + H(\vec{P})$,
	as desired.
\end{proof}
We previously said that the list of possible states in 
an eidostate represents a kind of information.  
Theorem~\ref{theorem:shannonentropy} puts this intuition on
a quantitative footing.  The entropy of a uniform eidostate
$E$ can be decomposed into two parts:  the average entropy
of the states, and an additional term representing the 
information contained in the distinction among the possible
states.  For a singleton state, the entropy is all of the first
sort.  For a pure information state $I \in \infoset$, it is all
of the second.

In fact, the decomposition itself uniquely picks out the entropic
probability assignment.  Suppose $P(a|E)$ is the entropic
probability of $a$ given $E$, and $P'(a|E)$ is some other probability
distribution over states in $E$.  By Gibbs's inequality \cite{MacKay2003},

\begin{equation}
	0 \leq \sum_{a \in E} P'(a|E) \log \left ( \frac{P'(a|E)}{P(a|E)} \right )
\end{equation}
with equality if and only if $P'(a|E) = P(a|E)$ for all $a \in E$.
We find that
\begin{eqnarray}
	0 & \leq &
	\sum_{a \in E} P'(a|E) \log \left ( \frac{P'(a|E) \, 2^{\entropy(E)}}{2^{\entropy(a)}} \right ) \nonumber \\[1ex]
	& = & \sum_{a \in E} P'(a|E) \log P'(a|E) \,\, + \,\,  \entropy(E) \,\, - \,\, \sum_{a \in E} P'(a|E) \entropy(a)
\end{eqnarray}
and so

\begin{equation}
	\entropy(E) \geq \biggl \langle \entropy(a) \biggr \rangle_{\!\! P'} + H(\vec{P}') ,
\end{equation}
with equality if and only if the $P'$ distribution is
the entropic one.

An agent that employs entropic probabilities will regard the
entropy of a uniform eidostate $E$ as the sum of two parts,
one the average entropy of the possible states and the other the
Shannon entropy of the distribution.  For an agent that employs
some other probability distribution, things will not be so simple.
Besides the average state entropy and the Shannon
entropy of the distribution, $\entropy(E)$ will include an
extra, otherwise unexplained term.  Thus, for a given collection
of eidostates connected by the $\rightarrow$ relation, the entropic
probability provides a uniquely simple account of
the entropy of any uniform eidostate.

It is in this sense we say that the entropic probability ``emerges'' from the
entropy function $\entropy$, just as that function itself emerges from the
$\rightarrow$ relation among eidostates.  

Some further remarks about probability are in order.
Every formal basis for probability emphasizes a distinct idea about it.
In Kolmogorov's axioms \cite{Kolmogorov1933}, probability is simply
a measure on a sample space.  High-measure subsets are more
probable.  In the Bayesian approach of Cox
\cite{Cox1961}, probability is a rational measure of confidence in 
a proposition.  Propositions in which a rational agent is more
confident are also more probable.
Laplace's early discussion \cite{Laplace1820} is
based on symmetry.  Symmetrically equivalent events---two
different orderings of a shuffled deck of cards, for instance---are
equally probable.  (Zurek \cite{Zurek2005} has used a similar
principle of ``envariance'' to discuss the origin of quantum probabilities.)
In algorithmic information theory \cite{CoverThomas}, the algorithmic 
probability of a bit string is related to its complexity.  Simpler bit strings 
are more probable.

In a similar way, entropic probabilities express facts about
state transformations.  In a uniform eidostate $E$, any two
states $a,b \in E$ are related by a possible process.  
If $a \rightarrow b$, then the output state is at least as
probable as the input state:  $P(a|E) \leq P(b|E)$.

\section{A Model for the Axioms}  \label{sec:model1}

A model for an axiomatic system may serve several purposes.
The existence of a model establishes that the axioms 
are self-consistent.  A model may also demonstrate that
the system can describe an actual realistic physical
situation.  If the axioms have a variety of
interesting models, then the axiomatic theory
is widely applicable.  We may almost say that
the entire significance of an axiomatic system
lies in the range of models for that system.

Terms that are undefined in an abstract axiomatic system
are defined within a model as particular mathematical
structures.  The axioms of the system are provable
properties of those structures.  Therefore, a model for
axiomatic information thermodynamics must include
several elements:
\begin{itemize}
	\item  A set $\stateset$ of states and a 
		collection $\eidoset$ of finite nonempty 
		subsets of $\stateset$ to be designated
		as eidostates.
	\item  A rule for interpreting the combination
		of states ($+$) in $\stateset$.
	\item  A relation $\rightarrow$ on $\eidoset$.
	\item  A designated set $\mechset \subseteq \stateset$
		of mechanical states (which might be empty).
	\item  Proofs of Axioms~\ref{axiom:eidostates}--\ref{axiom:equivalence}
		within the model, including the general properties of $\rightarrow$,
		the existence of record states and information states,
		etc.
\end{itemize}

The model will therefore involve specific meanings for $\stateset$, 
$\eidoset$, $+$, $\rightarrow$ and so forth.  It will also yield 
interpretations of derived concepts and results, such as entropy
functions and conserved components of content. In the abstract
		theory, the combination $a+b$ is simply
		the Cartesian pair $(a,b)$.  This definition may suffice
		for the model, or the model may have a 
		different interpretation of $+$.  In any case,
		it must be true in the model that $a+b=a'+b'$
		implies $a=a'$ and $b=b'$.

Our first model for axiomatic thermodynamics
is based on a set $\atomicset$ of ``atomic'' states, from which all states 
in $\stateset$ and all eidostates in $\eidoset$ are constructed.  We 
assign entropies and components of content to these states, which
extend to composite states by additivity.  Let us consider a simple but
non-trivial example that has just one component of content $Q$.
A suitable set $\atomicset$ of atomic states is shown in 
Figure~\ref{fig:atomicstates}.  It includes a special state $r$
(with $\entropy(r) = 0$ and $Q(r) = 0$) and a continuous set of states
$s_{\lambda}$ with $Q(s_{\lambda}) = 1$ and $\entropy(s_{\lambda}) = \lambda$.
The parameter $\lambda$ ranges over the closed interval $[0,1]$.
\begin{figure}
\begin{center}
\includegraphics[width=2in]{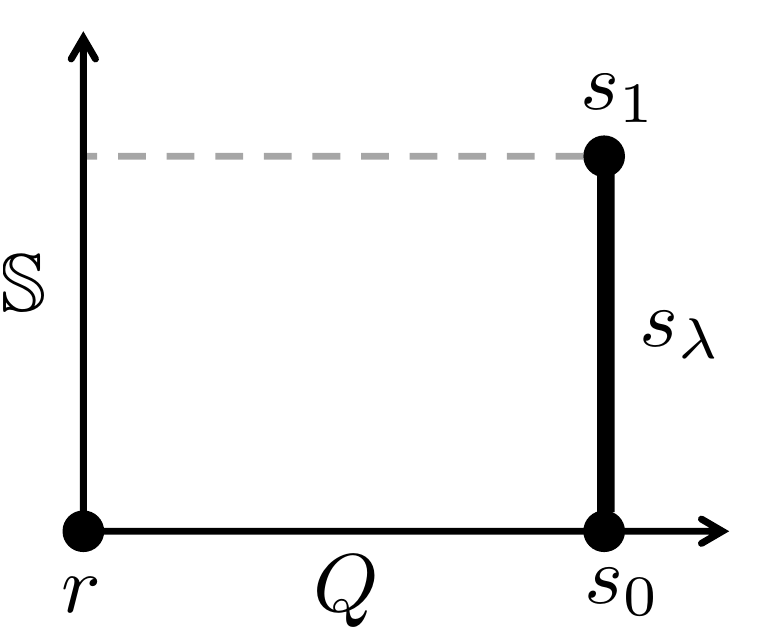}
	\end{center}
	\caption{\label{fig:atomicstates} Atomic states in our
		simple ``macrostate'' model.}
\end{figure}
The set of states $\stateset$ includes everything that can be 
constructed from $\atomicset$ by finite application of the 
pairing operation.  In this way, we can build up $a \in \stateset$
with any non-negative integer value of the component of content
$Q(a)$ and any entropy value $0 \leq \entropy(a) \leq Q(a)$.
Indeed, there will typically be many different ways to create
given $(Q, \entropy)$ values.  To obtain $Q(a) = 2$ 
and $\entropy(a)=1$, for example, 
we might have $a = s_{1}+s_{0}, s_{0}+s_{1}, (s_{1/2}+s_{1/2})+r, \ldots$.

Anticipating somewhat, we call an eidostate \emph{uniform} if all of 
its elements have the same $Q$-value.  We calculate the entropy
of a uniform eidostate by applying Theorem~\ref{theorem:entropyformula}
to it.

Our model allows any finite nonempty set of states
to play the role of an eidostate.  Hence,
we have both uniform and non-uniform eidostates in $\eidoset$.
Each eidostate $A$ has a finite Cartesian factorization

\begin{equation}
	A = F_{A}( E_{1}, \ldots , E_{n}) .
\end{equation}
If $A$ is uniform, then all of its factors are also uniform.  If
none of its factors are uniform, we say that $A$ is
\emph{completely non-uniform}.  More generally, we can
write down an \emph{NU-decomposition} for any eidostate $A$:

\begin{equation}
	A \sim N_{A} + U_{A},
\end{equation}
where $N_{A}$ is a completely non-uniform eidostate
and $U_{A}$ is a uniform eidostate.  Of course, 
if $A$ itself is either completely non-uniform or uniform, 
one or the other of these eidostates may be absent 
from the decomposition.  The NU-decomposition is 
unique up to similarity:  If $N_{A} + U_{A} \sim N_{A}' + U_{A}'$
for completely non-uniform $N$s and uniform $U$s, 
then $N_{A} \sim N_{A}'$ and $U_{A} \sim U_{A}'$.

We can now define the $\rightarrow$ relation on $\eidoset$
in our model.  If $A,B \in \eidoset$, we first write down
NU-decompositions $A \sim N_{A} + U_{A}$ and 
$B \sim N_{B} + U_{B}$.  We  say that $A \rightarrow B$
provided three conditions~hold:
\begin{enumerate}
	\item  Either $N_{A}$ and $N_{B}$ both do not exist,
		or $N_{A} \sim N_{B}$.
	\item  Either $U_{A}$ and $U_{B}$ both do not exist, 
		or only one exists and its $Q$-value is 0, or
		both exist and $Q(U_{A}) = Q(U_{B})$.
	\item  Either $U_{A}$ and $U_{B}$ both do not exist,
		or only $U_{A}$ exists and $\entropy(U_{A}) = 0$,
		or only $U_{B}$ exists and $\entropy(U_{B}) \geq 0$,
		or both exist and $\entropy(U_{A}) \leq \entropy(U_{B})$.
\end{enumerate}
We may call these the $N$-criterion, $Q$-criterion,
and $\entropy$-criterion, and summarize their meaning
as follows:  
$A \rightarrow B$ provided we can transform
$A$ to $B$ by: (a) rearranging the non-uniform factors; 
and (b) transforming the uniform
factors in a way that conserves $Q$ and does not
decrease $\entropy$.

Now, let us examine each of the axioms in turn.
\begin{description}
	\item[Axiom~\ref{axiom:eidostates}]  The basic
		properties of eidostates follow by construction.
	\item[Axiom~\ref{axiom:processes}]  Part (a) holds because
		$A \sim B$ implies that $A$ and $B$ can have the
		same NU-decomposition.  Part (b) holds because
		similarity, equality (for $Q$) and inequality (for $\entropy$)
		are all transitive.  Parts (c) and (d) make use of the 
		general facts that $N_{A+B} \sim N_{A} + N_{B}$ and
		$U_{A+B} \sim U_{A}+U_{B}$.
	\item[Axiom~\ref{axiom:austin}]  If $N_{A} \not \sim N_{B}$, then
		$A \nrightarrow B$.  If $N_{A} \sim N_{B}$, then it must
		be true that $U_{B} \subsetneq U_{A}$, and so 
		$\entropy(U_{A}) > \entropy(U_{B})$.  The 
		$\entropy$-criterion fails, so $A \nrightarrow B$ in this 
		case as well.
	\item[Axiom~\ref{axiom:conditional}]   For Part (a), we note that
		$A$ must be uniform, and so $A' \subseteq A$ is also uniform.
		The statement follows from the $\entropy$-criterion.  Part (b)
		also follows from the $\entropy$-criterion.
	\item[Axiom~\ref{axiom:information}]  The atomic state $r$ with $Q(r)=0$
		and $\entropy(r) = 0$ is a record state, as is $r+r$, etc.  We can
		take our bit state to be $\bitstate = \{ r, r+r \}$.  Since every information state
		$I \in \infoset$ is uniform with $Q(I) = 0$, every information process
		(including $\bitproc = \fproc{r}{\bitstate}$) is possible.
	\item[Axiom~\ref{axiom:demons}]  Since all of the states of the form $a + I$
		are uniform, the statements in this axiom follow from the $\entropy$-criterion.
	\item[Axiom~\ref{axiom:stability}]  Suppose $nA \rightarrow nB + J$.  Since
		$J$ is uniform, it must be that $nN_{A} \sim nN_{B}$, from which it
		follows that $N_{A} \sim N_{B}$.  The $\entropy$-criterion for $U_{A}$
		and $U_{B}$ follows from a typical stability argument---that is, if
		$nx \leq ny + z$ for arbitrarily large values of $n$, then it must
		be true $x \leq y$.
	\item[Axiom~\ref{axiom:mechanical}]  The set $\mechset$ of mechanical
		states may be defined to include all states that can be constructed
		from the zero-entropy atomic state $s_{0}$ (such as 
		$s_{0}+s_{0}$, $s_{0}+(s_{0}+s_{0})$, etc.).  
		The required properties of $\mechset$~follow.
	\item[Axiom~\ref{axiom:equivalence}]  The uniform eidostate $E$ has 
		$Q(E) = q \geq 0$ and $\entropy(E) = \sigma \geq 0$.  
		Choose an integer $n > \sigma$.
		Now, let $e = q s_{0}$ (or $e = r$ if $q = 0$), $x = n s_{0}$,
		and $y = n s_{\lambda}$ where $\lambda = \sigma/n$.
		We find that $Q(e) = q$, $Q(x) = Q(y) = n$, 
		$\entropy(e) = \entropy(x) = 0$ and $\entropy(y) = n (\sigma/n) = \sigma$.
		It follows that $x \rightarrow y$ and $E + x \leftrightarrow e + y$.
\end{description}

This model based on a simple set of atomic states has several 
sophisticated characteristics, including a non-trivial component of content 
$Q$ and possible processes involving non-uniform eidostates.  
It is not difficult to create models of this type that are even richer 
and more complex.  However, it may be objected that this type of model
obscures one of the key features of axiomatic information thermodynamics.
Here, the entropy function $\entropy$ does not \emph{emerge from} the $\rightarrow$ 
relation among eidostates, but instead is imposed by hand to
\emph{define} $\rightarrow$ within the model.  We address this 
deficiency in our next model.

\section{A Simple Quantum Model}  \label{sec:model2}

Now, we present a model for the axioms in which
the entropy function does emerge from the underlying structure.
The model is a simple one without mechanical states or non-trivial
components of content.  Every eidostate is uniform and every
process is possible.  On the other hand, the model is based on
quantum mechanics, and so is not devoid of features of interest.

Consider an agent A that can act upon an external qubit system 
Q having Hilbert space $\mathcal{Q}$.
Based on the information the agent possesses, it may assign the
states $\ket{\psi_{1}}$ or $\ket{\psi_{2}}$ to the qubit.  These
state vectors need not be orthogonal. 
That is, it may be that no measurement of Q can perfectly 
distinguish which of the two states is actually present.  
The states, however, correspond to states of knowledge of agent A,
and the agent is able to perform different operations on Q depending
on whether it judges the qubit to be in one state or the other. Our
notion of information possessed by the agent is thus similar to
Zurek's concept of \emph{actionable information} \cite{Zurek2013}.  Roughly speaking,
information is actionable if it can be used as a control variable for
conditional unitary dynamics.
This means that the two states of the agent's memory 
($\ket{\mu_{1}}$ and $\ket{\mu_{2}}$ in a Hilbert space $\mathcal{A}$) 
must be distinguishable.
Hence, if we include the agent in our description of the entire system, 
the states $\ket{\mu_{1}} \otimes \ket{\psi_{1}}$ and 
$\ket{\mu_{2}} \otimes \ket{\psi_{2}}$ are orthogonal,
even if $\ket{\psi_{1}}$ and $\ket{\psi_{2}}$ are not.

Our model for axiomatic information thermodynamics envisions
a world consisting of an agent A and an unbounded number of
external qubits. Nothing essential in our model would
be altered if the external systems had $\dim \mathcal{Q} = 
d$ ---``qudits'' instead of qubits. The thermodynamic states
of the qubit systems are actually states of knowledge of the agent,
and so we must include the corresponding state of the agent's 
memory in our physical description.  The quantum
state space for our model is of the~form:

\begin{equation}
	\hilbert = \mathcal{A} \otimes \mathcal{Q} \otimes \mathcal{Q} \otimes \cdots
\end{equation}
To be a bit more rigorous, we restrict $\hilbert$ to vectors of
the form $\ket{\Psi} \otimes \ket{0} \otimes \ket{0} \otimes \cdots$,
where $\ket{\Psi} \in \mathcal{A} \otimes \mathcal{Q}^{\otimes n}$
for some finite $n$, and $\ket{0}$ is a designated ``zero'' ket in $\mathcal{Q}$.
Physical states in $\hilbert$ have $\amp{\Psi}{\Psi} = 1$.
(The space $\hilbert$ is not quite a Hilbert space, since it is 
not topologically complete, but this mathematical nicety will not
affect our discussion.)

Since the Qs are qubits, $\dim \mathcal{Q} = 2$.  
The agent space
$\mathcal{A}$, however, must be infinite-dimensional, so 
that it contains a countably infinite set of 
orthogonal quantum states.  These are to be identified 
as distinct records of the agent's memory.

In our thermodynamic model, the elements of $\stateset$
(the thermodynamic ``states'') are projection operators
on $\hilbert$.  For any $a \in \stateset$, we have
a projection on $\hilbert$ of the form
\begin{equation}
	\oper{\Pi}_{a} = \proj{a} \otimes \oper{\pi}_{a} \otimes \proj{0} \otimes \cdots 
\end{equation}
where $\ket{a}$ is an agent state in $\mathcal{A}$ and $\oper{\pi}_{a}$ is a
non-null projection in $\mathcal{Q}^{\otimes n}$ for some finite $n \geq 1$.
The value of $n$ is determined by a specified integer function $L(a)$, which
we call the \emph{length} of the state $a$. 
Heuristically, the thermodynamic state $a$ means that the state
of the world lies in the subspace $\mathcal{S}_{a}$ onto which
$\oper{\Pi}_{a}$ projects.  The agent's memory is in the state 
$\ket{a}$ and the the quantum state of the
first $L(a)$ external qubits lies somewhere in the subspace 
onto which $\oper{\pi}_{a}$ projects.  
(All of the subsequent qubits are in the state $\ket{0}$.)

Two distinct thermodynamic states correspond to
orthogonal states of the agent's memory.  If $a,b \in \stateset$
with $a \neq b$, then $\amp{a}{b} = 0$.  The projections
$\oper{\Pi}_{a}$ and $\oper{\Pi}_{b}$ are orthogonal to 
each other (so that $\oper{\Pi}_{a} \oper{\Pi}_{b} = 0$).
However, it need not be the case that $\oper{\pi}_{a}$
and $\oper{\pi}_{b}$ are orthogonal.

Given $a$, the projection $\oper{\Pi}_{a}$ projects onto 
the subspace $\mathcal{S}_{a}$  The dimension of this
subspace is $d_{a} = \dim \mathcal{S}_{a} = \tr \oper{\Pi}_{a}
= \tr \oper{\pi}_{a}$.   Note that $d_{a} \leq 2^{L(a)}$.
We will assume that there are $\mathcal{S}_{a}$ subspaces
of every finite dimension:  For any integer $n \geq 1$, there
exists $a \in \stateset$ with $d_{a} = n$.

Suppose $a,b \in \stateset$ correspond to 
$\oper{\Pi}_{a} = \proj{a} \otimes \oper{\pi}_{a} \otimes \cdots$
and $\oper{\Pi}_{b} = \proj{b} \otimes \oper{\pi}_{b} \otimes \cdots$.
Then, we will specify that the combined state $a + b$ corresponds to 
\begin{equation}
	\oper{\Pi}_{a+b} = \proj{a+b} \otimes \oper{\pi}_{a}
		\otimes \oper{\pi}_{b} \otimes \cdots \,\,.
\end{equation}
Since the state $a+b$ entails a distinct state of the agent's
knowledge, the agent state vector $\ket{a+b}$ is orthogonal 
to both $\ket{a}$ and $\ket{b}$.
We also note that $L(a+b) = L(a)+L(b)$.

Here is a clarifying example.  Suppose $a,b,c \in \stateset$.  Then,
\begin{eqnarray}
	\oper{\Pi}_{(a+b)+c} & = & \proj{(a+b)+c} \otimes
		\oper{\pi}_{a} \otimes \oper{\pi}_{b} \otimes \oper{\pi}_{c} \otimes \cdots \\
	\oper{\Pi}_{a+(b+c)} & = & \proj{a+(b+c)} \otimes
		\oper{\pi}_{a} \otimes \oper{\pi}_{b} \otimes \oper{\pi}_{c} \otimes \cdots 
\end{eqnarray}
are distinct thermodynamic states and hence orthogonal projections in $\hilbert$,
even though they correspond to exactly the same qubit states.  
The difference between $(a+b)+c$ and $a+(b+c)$ entirely lies 
in the distinct representations of the states in the agent's memory.

The eidostates in our model are the finite nonempty collections
of states in $\stateset$.  We can associate each eidostate with a 
projection operator as well.  Let $E = \{ a, \ldots \}$ be an
eidostate.  We define

\begin{equation}
	\oper{\Pi}_{E} = \sum_{a \in E} \oper{\Pi}_{a} 
		= \sum_{a \in E} \proj{a} \otimes \oper{\pi}_{a} \otimes \cdots .
\end{equation}
This is a projection operator because the $\oper{\Pi}_{a}$
projections are orthogonal to one another.  $\oper{\Pi}_{E}$ projects onto
a subspace $\mathcal{S}_{E}$, which is the linear span
of the collection of subspaces $\{ \mathcal{S}_{a}, \ldots \}$.

Interestingly, this subspace $\mathcal{S}_{E}$ might contain
quantum states in which the agent A is entangled with one
or more external qubits.  Suppose $a, b \in \stateset$ are associated
with single-qubit projections onto distinct states $\ket{\psi_{a}}$ and $\ket{\psi_{b}}$.
The eidostate $E = \{ a, b \}$ is associated with the projection

\begin{equation}
	\oper{\Pi}_{E} = \left ( \proj{a} \otimes \proj{\psi_{a}} + \proj{b} \otimes \proj{\psi_{b}} \right )
		\otimes \cdots ,
\end{equation}
which projects onto a subspace $\mathcal{S}_{E}$ that contains the
quantum state

\begin{equation}
	\ket{\Psi} = \frac{1}{\sqrt{2}} \left ( \ket{a} \otimes \ket{\psi_{a}} + \ket{b} \otimes \ket{\psi_{b}} \right )
		\otimes \cdots .
\end{equation}
In this state, the agent does not have a definite memory record
state.  However, if a measurement is performed on the agent
(perhaps by asking it a question), then the resulting memory record
$a$ or $b$ would certainly be found to be consistent with the state 
of the qubit system, $\ket{\psi_{a}}$ or $\ket{\psi_{b}}$.

Suppose we combine two eidostates $A = \{ a, \ldots \}$ and $B = \{ b, \ldots \}$.
Then, the quantum state lies in a subspace of dimension
\begin{eqnarray}
	d_{A+B} = \tr \oper{\Pi}_{A+B} 
	 & = & \tr \sum_{a,b} \oper{\Pi}_{a+b} \nonumber \\
	 & = & \sum_{a,b} \tr \left ( \proj{a+b} \otimes \oper{\pi}_{a} \otimes \oper{\pi}_{b} \right ) \nonumber \\
	 & = & \sum_{a,b} \left ( \tr \oper{\pi}_{a} \right ) \left ( \tr \oper{\pi}_{b} \right ) \nonumber \\
	 & = & \left ( \sum_{a} \tr \oper{\pi}_{a} \right ) \left ( \sum_{b} \tr \oper{\pi}_{b} \right ) \nonumber \\
	 & = & d_{A} \cdot d_{B} .
\end{eqnarray}
When eidostates combine, subspace dimension is multiplicative.

It remains to define the $\rightarrow$ relation in our quantum model.
We say that $A \rightarrow B$ if there exists a unitary time evolution
operator $\oper{U}$ on $\hilbert$ such that $\ket{\Psi} \in \mathcal{S}_{A}$
implies that $\oper{U} \ket{\Psi} \in \mathcal{S}_{B}$.  That is, every
quantum state consistent with $A$ evolves to one consistent with $B$
under the time evolution $\oper{U}$.  (Note that the evolution includes
a suitable updating of the agent's own memory state.)  
This requirement is easily expressed as a subspace dimension 
criterion:  $A \rightarrow B$ if and only if $d_{A} \leq d_{B}$.

We are now ready to verify our axioms.
\begin{description}
	\item[Axiom~\ref{axiom:eidostates}]  This follows from our 
		construction of the eidostates $\eidoset$.
	\item[Axiom~\ref{axiom:processes}]  All of these basic properties
		of the $\rightarrow$ relation follow from the subspace
		dimension criterion.
	\item[Axiom~\ref{axiom:austin}]  If $B \subsetneq A$, then
		$d_{A} > d_{B}$, and so $A \nrightarrow B$.
	\item[Axiom~\ref{axiom:conditional}]  Again, both parts of this axiom
		follow from the subspace dimension criterion.  If eidostate $A$
		is a disjoint union of eidostates $A_{1}$ and $A_{2}$, then
		$d_{A} = d_{A_{1}} + d_{A_{2}}$.
	\item[Axiom~\ref{axiom:information}]  Any state $r$ with $d_{r} = 1$
		functions as a record state.  We have assumed that such a 
		state exists.  We can take $\bitstate = \{ r, r+r \}$.  The bit
		process $\bitproc = \fproc{r}{\bitstate}$ is natural ($r \rightarrow \bitstate$)
		by the subspace dimension criterion.  Notice that, for any information state
		$I$, $d_{I} = \numberin{I}$.
	\item[Axiom~\ref{axiom:demons}]  For any $b \in \stateset$ and $I \in \infoset$,
		we have $d_{b+I} = d_{b} \cdot \numberin{I}$.  For Part (a), we can always
		find a large enough information state so that $d_{b} \leq d_{a} \cdot \numberin{I}$.
		For Part (b), either $d_{a} \leq d_{b + I}$ or $d_{b+I} \leq d_{a}$.
	\item[Axiom~\ref{axiom:stability}]  If $(d_{A})^{n} \leq (d_{B})^{n} \cdot \numberin{J}$ for
		arbitrarily large values of $n$, then $d_{A} \leq d_{B}$.
	\item[Axiom~\ref{axiom:mechanical}]  It is consistent to take $\mechset = \emptyset$.
	\item[Axiom~\ref{axiom:equivalence}]  All of our eidostates are uniform.  For any eidostate
		$E$, we can choose $e$ so that $d_{e} = d_{E}$.  (Recall that we have assumed
		states with every positive subspace dimension.)  If we chose $x=y$ to be any
		state, then $E + x \leftrightarrow e + y$.
\end{description}

Our quantum mechanical model is therefore a model of the
axioms of information thermodynamics.  It is a relatively
simple model, of course, having no non-trivial conserved 
components of content and no mechanical states.

In the quantum model, the entropy of any eidostate is simply 
the logarithm of the dimension of the corresponding 
subspace:  $\entropy(E) = \log d_{E}$.  This is the
von Neumann entropy of a uniform density operator
$\oper{\rho}_{E} = \frac{1}{d_{E}} \oper{\Pi}_{E}$.
We can, in fact, recast our entire discussion
in terms of these mixed states, and this approach
does yield some insights.  For example, we find that 
the density operator $\oper{\rho}_{E}$ for eidostate $E$ 
is a mixture of the density operators for its constituent states:

\begin{equation}
	\oper{\rho}_{E} = \sum_{a \in E} P(a|E) \oper{\rho}_{a} ,
\end{equation}
where $P(a|E)$ is the entropic probability

\begin{equation}
	P(a|E) = \frac{2^{\entropy(a)}}{2^{\entropy(E)}} = \frac{d_{a}}{d_{E}} .
\end{equation}

There are, of course, many quantum states of the agent and its
qubit world that do not lie within any eidostate subspace 
$\mathcal{S}_{E}$.  For example, consider a state 
associated with a pure state projection 
$\oper{\Pi}_{a} = \proj{a} \otimes \proj{\psi_{a}} \otimes \cdots$.
Let $\ket{\chi}$ be a state orthogonal to $\ket{\psi_{a}}$.
Then, $\ket{a} \otimes \ket{\chi} \otimes \cdots$ is a
perfectly legitimate quantum state that is orthogonal to 
$\mathcal{S}_{a}$ and every other eidostate subspace.
This state represents a situation in which the agent's 
memory record indicates that the first $L(a)$ external qubits 
are in state $\ket{\psi_{a}}$, but \emph{the agent is wrong.}

The exclusion of such physically possible but incongruous
quantum states tells us something significant about our
theory of axiomatic information thermodynamics.  
The set $\eidoset$ does not necessarily
include all possible physical situations; the arrow relations
$\rightarrow$ between eidostates do not necessarily represent
all possible time evolutions.  Our axiomatic system is
simply a theory of what transformations are possible
among a collection of allowable states.  In this, it is
similar to ordinary classical thermodynamics, which is
designed to consider processes that begin and end 
with states in internal thermodynamic equilibrium.

\section{Remarks}  \label{sec:remarks}
The emergence of the entropy $\entropy$, a state function 
that determines the irreversibility of processes, is a key benchmark 
for any axiomatic system of thermodynamics.  
Our axiomatic system does not yield a unique entropy
on $\stateset$, since it is based on the extension
of an irreversibility function to impossible processes.  
However, many of our results
and formulas for entropy are uniquely determined by our 
axioms.  The entropy measure for information states, 
the Hartley--Shannon
entropy $\log \numberin{I}$, is unique up to
the choice of logarithm base.  
This in turn uniquely determines the irreversibility function
on possible singleton processes, since this is defined in terms of
the creation and erasure of bit states.
There is a unique relationship 
between the entropy of a uniform eidostate and the entropy 
of the possible states it contains.  Finally, the entropic
probability distribution on a uniform eidostate, which might appear
at first to depend on the singleton state entropy, is
nonetheless unique. 

It remains to be seen how the axiomatic system 
developed here for state transformations is related to the axiomatic system, 
similar in some respects, given by Knuth and Skilling for considering 
problems of inference~\cite{KnuthSkilling2012}.  There, symmetry axioms 
in a lattice of states give rise to probability and entropy measures.  
In a similar way, the "entropy first, probability after" idea presented here
is reminiscent of Caticha's "entropic inference"~\cite{Caticha2011}, in which 
the probabilistic Bayes rule is derived (along with the maximum entropy 
method) from a single relative entropy functional.

The entropy in information thermodynamics has two aspects: 
a measure of the information in a pure information state and 
an irreversibility measure for a singleton eidostate (the 
kind most closely analogous to a conventional thermodynamic state).
The most general expression for the entropy of a
uniform eidostate in Theorem~\ref{theorem:shannonentropy}
exhibits this twofold character.
However, the two aspects of entropy are not really distinct
in our axiomatic system. 
Both are based on the structure of the
$\rightarrow$ relation among eidostates,
which tells how one eidostate may be transformed into another 
by processes that may include demons.

The connection between information and thermodynamic
entropy is nowhere more clearly stated than in Landauer's
principle, the minimum thermodynamic cost of information erasure.
It is easy to state a theorem of our system corresponding
to Landauer's principle.  Suppose $a,b \in \stateset$ and
$\bitstate$ is a bit state.  If $a + \bitstate \rightarrow b$, then
it follows that $\entropy(b) \geq \entropy(a) + 1$.  Erasing
a bit state is necessarily accompanied by an increase of 
at least one unit in the thermodynamic entropy.  More 
generally, suppose $A$ and $B$ are uniform eidostates
with $A \rightarrow B$.  If we use Theorem~\ref{theorem:shannonentropy}
to write the entropies of the two states as

\begin{equation}
	\entropy(A) = \langle \entropy \rangle_{A} + H_{A}
	\qquad \mbox{and} \qquad
	\entropy(B) = \langle \entropy \rangle_{B} + H_{B},
\end{equation}
then it follows that $\Delta \langle \entropy \rangle \geq - \Delta H$
for the process taking $A$ to $B$.  In other words, any decrease
of the ``Shannon information'' part of the eidostate entropy 
($- \Delta H$) must be accompanied by an increase, 
on average, in the thermodynamic entropy of the possible states
($\Delta \langle \entropy \rangle$).

These theorems do not really constitute
a ``proof'' of Landauer's principle, because they
depend upon the physical applicability of our axioms.
On the other hand, our axiomatic framework does
obviate some of the objections that have been raised
to existing derivations of Landauer's principle.
Norton \cite{Norton2005}, for example, has argued that many ``proofs''
of the principle improperly mix together two different
kinds of probability distribution, the microstate 
distributions associated with thermodynamic 
equilibrium states and the probability assignments
to memory records that represent information.
He calls this the problem of ``illicit ensembles''.
Our approach, by contrast, is not based on concepts of
probability, except for those that emerge naturally
from the structure of the $\rightarrow$ relation.
We do not resolve thermodynamic states into their
microstates at all, and we represent information by 
a simple enumeration of possible memory records.

The Second Law of Thermodynamics, the law of non-decrease
of entropy, is the canonical example of the ``arrow of time''
in physics.  Time asymmetry in our axiomatic theory is
found in the direction of the $\rightarrow$ relation.  It is
an enlightening exercise to consider the axioms of information
thermodynamics with the arrows reversed.  Some axioms
(e.g., Axioms~\ref{axiom:eidostates} and \ref{axiom:processes})
are actually unchanged by this.  Others may be modified but
still remain true statements in the theory.  A few become
false.  The most striking example of the last is Axiom~\ref{axiom:austin},
which states that no eidostate may be transformed into a 
proper subset of itself.  That is, no process can deterministically
delete one of a list of possible states.  This is the principle that leads to 
irreversbility in information processes, and through them,
to more general irreversibility and entropy measures.

\bigskip

The authors are pleased to acknowledge many useful
conversations with, and suggestions from, Rob Spekkens,
Jonathan Oppenheim, David Jennings, Charles Bennett, 
John Norton, Graham Reid, Avery Tishue, Ian George, Aixin Li, 
and Tiancheng Zheng.  AH was supported
by the Kenyon Summer Science Scholar program in the 
summer of 2016.  BWS and MDW gratefully acknowledge
funding from the Foundational Questions Institute (FQXi),
via grant FQXi-RFP-1517.

\section*{Appendix}
This appendix provides a convenient summary of 
axioms of our theory.  

First, we remind ourselves of a 
few essential definitions:

An \textbf{\emph{eidostate}} is a set whose elements are called 
\textbf{\emph{states}}.
The collection of eidostates is $\eidoset$ and the
collection of states is $\stateset$.  An element $a \in \stateset$
may be identified with the singleton eidostate $\{ a \} \in \eidoset$.

Eidostates are combined by the Cartesian product, which we
denote by the symbol $+$.  When we combine an eidostate with itself $n$ times,
we use $nA$ to denote $A + (A + (A + \ldots))$.  
Two eidostates are \textbf{\emph{similar}}
(written $A \sim B$) if they are made up of the same 
Cartesian factors, perhaps combined in a different way.

There is a relation $\rightarrow$ on $\eidoset$, and thus also
on the singletons in $\stateset$.
An eidostate $A$ is \textbf{\emph{uniform}} if, for all $a,b \in A$, 
either $a \rightarrow b$ or $b \rightarrow a$.
A formal \textbf{\emph{process}} is a pair of eidostates $\fproc{A}{B}$.  We say
that a process $\fproc{A}{B}$ is \textbf{\emph{possible}} if either $A \rightarrow B$
or $B \rightarrow A$.  

A \textbf{\emph{record state}} $r$ is a state for which there exists another state
$a$ such that $a \rightarrow a+r$ and $a+r \rightarrow a$ (denoted
$a \leftrightarrow a+r$).  An \textbf{\emph{information state}} is an eidostate
containing only record states, and the set of information states is
called $\infoset$.  A \textbf{\emph{bit state}} $\bitstate$ is an information state 
containing exactly two distinct record states.
A \textbf{\emph{bit process}} is a formal process
$\fproc{r}{\bitstate}$, where $r$ is a record state and $\bitstate$
is a bit state.

Now, we may state our axioms:

\begin{description}
\item[Axiom~\ref{axiom:eidostates}] (Eidostates.)
	$\eidoset$ is a collection of sets called \emph{eidostates} such that:
	\begin{description}
		\item[(a)]  Every $A \in \eidoset$ is a finite nonempty set
			with a finite prime Cartesian factorization.
		\item[(b)]  $A + B \in \eidoset$ if and only if $A,B \in \eidoset$.
		\item[(c)]  Every nonempty subset of an eidostate is also an eidostate.
	\end{description}
\item[Axiom~\ref{axiom:processes}] (Processes.)
Let eidostates $A,B,C \in \eidoset$, and $s \in \stateset$.
\begin{description}
	\item[(a)]  If $A \sim B$, then $A \rightarrow B$.
	\item[(b)]  If $A \rightarrow B$ and $B \rightarrow C$, then $A \rightarrow C$.
	\item[(c)]  If $A \rightarrow B$, then $A + C \rightarrow B + C$.
	\item[(d)]  If $A + s \rightarrow B + s$, then $A \rightarrow B$.
\end{description}
\item[Axiom~\ref{axiom:austin}]
	If $A,B \in \eidoset$ and $B$ is a proper subset of $A$, then $A \nrightarrow B$.
\item[Axiom~\ref{axiom:conditional}] (Conditional processes.)
	\begin{description}
		\item[(a)]  Suppose $A, A' \in \eidoset$ and $b \in \stateset$.  If $A \rightarrow b$
			and $A' \subseteq A$ then $A' \rightarrow b$.
		\item[(b)]  Suppose $A$ and $B$ are uniform eidostates that are each disjoint
			unions of eidostates: $A = A_{1} \cup A_{2}$ and $B = B_{1} \cup B_{2}$.
			If $A_{1} \rightarrow B_{1}$ and $A_{2} \rightarrow B_{2}$ then
			$A \rightarrow B$.
	\end{description}
\item[Axiom~\ref{axiom:information}]  (Information.)
	There exist a bit state and a possible bit process.
\item[Axiom~\ref{axiom:demons}]  (Demons.)
	Suppose $a,b \in \stateset$ and $J \in \infoset$ such that $a \rightarrow b + J$.
	\begin{description}
		\item[(a)]  There exists $I \in \infoset$ such that $b \rightarrow a+I$.
		\item[(b)]  For any $I \in \infoset$, either $a \rightarrow b + I$ or $b+I \rightarrow a$.
	\end{description} 
\item[Axiom~\ref{axiom:stability}]  (Stability.)
	Suppose $A,B \in \eidoset$ and $J \in \infoset$.  If $nA \rightarrow nB + J$ for
	arbitrarily large values of $n$, then $A \rightarrow B$.
\item[Axiom~\ref{axiom:mechanical}]  (Mechanical states.)
	There exists a subset $\mechset \subseteq \stateset$ of 
	\emph{mechanical states} such that:
	\begin{description}
		\item[(a)]  If $l,m \in \mechset$, then $l+m \in \mechset$.
		\item[(b)]  For $l,m \in \mechset$, if $l \rightarrow m$
			then $m \rightarrow l$.
	\end{description}	
\item[Axiom~\ref{axiom:equivalence}]  (State equivalence.)
	If $E$ is a uniform eidostate then there exist states $e,x,y \in \stateset$
	such that $x \rightarrow y$ and $E + x \leftrightarrow e + y$.
\end{description}

\bibliographystyle{unsrt}
\bibliography{quantum}

\end{document}